\pdfoutput=1

\documentclass[11pt]{article}

\usepackage[margin=1in]{geometry}

\usepackage{amsmath,amssymb,amsthm}
\usepackage{graphicx}
\usepackage{xcolor}
\usepackage{hyperref}
\usepackage{cleveref}
\usepackage{enumerate}
\usepackage{comment}

\usepackage{bbm}
\usepackage{nicefrac}
\usepackage{xfrac}
\usepackage{xspace}

\usepackage[ruled]{algorithm}
\usepackage{algpseudocode}

\usepackage[T1]{fontenc}
\usepackage{lmodern}

\newcommand{\MLOPlong}{{Minimum Linear Ordering Problem}\xspace}
\newcommand{\MLOP}{{MLOP}\xspace}
\newcommand{\SBClong}{{Minimum Submodular $b$-Balanced Cut}\xspace}
\newcommand{\SBC}{{MSBC}\xspace}
\newcommand{\MSSClong}{{Minimum Sum Set Cover}\xspace}
\newcommand{\MSVClong}{{Minimum Sum Vertex Cover}\xspace}
\newcommand{\MLAlong}{{Minimum Linear Arrangement}\xspace}
\newcommand{\MCIGlong}{{Minimum Containing Interval Graph}\xspace}
\newcommand{\MSlong}{{Minimum SumCut}\xspace}
\newcommand{\MPlong}{{Minimum Profile}\xspace}
\newcommand{\MLSClong}{{Minimum Latency Set Cover}\xspace}
\newcommand{\WSSClong}{{Weighted Submodular Sparsest Cut}\xspace}
\newcommand{\WSSC}{{WSSC}\xspace}

\newcommand{\ind}{\mathbf{1}}
\newcommand{\OPT}{\operatorname{OPT}}
\newcommand{\val}{\operatorname{val}}
\newcommand{\hS}{h}
\newcommand{\Alg}{\operatorname{Alg}}

\newtheorem{theorem}{Theorem}
\newtheorem{lemma}[theorem]{Lemma}
\newtheorem{corollary}[theorem]{Corollary}

\theoremstyle{definition}

\theoremstyle{remark}

\newtheorem*{remark*}{Remark}

\numberwithin{equation}{section}

\title{Tight Algorithm and Hardness for Submodular Linear Ordering}

\author{
    Evan Abboud\textsuperscript{1} \qquad
    Roy Schwartz\textsuperscript{1} \\[2ex]
    \textsuperscript{1}Technion -- Israel Institute of Technology
}

\date{}

\begin{document}

\maketitle

\begin{abstract}
We consider the Minimum Linear Ordering Problem: given a ground set $N$ of cardinality $n$ and a non-negative set function $f\colon 2^N\rightarrow \mathbb{R}_{\geq 0}$, the goal is to find an ordering $\pi$ of $N$ that minimizes the sum of the values of $f$ over all prefixes of $\pi$.
This problem has been studied for various classes of set functions, and the case of a submodular $f$ is of special interest, as it captures classic problems including Minimum Linear Arrangement and Minimum Containing Interval Graph.
In this work, we resolve the approximability of the Minimum Linear Ordering Problem for a general submodular $f$ by establishing matching upper and lower bounds and present:
$(1)$ a polynomial-time algorithm achieving an $O(\sqrt{n/\ln n})$-approximation; and
$(2)$ a matching information-theoretic hardness result, showing that no algorithm evaluating $f$ a polynomial number of times can achieve an $o(\sqrt{n/\ln n})$-approximation.
Previously, the best known hardness of approximation was $2$, and an $O(\sqrt{n/\ln n})$-approximation was known only for the special case where $f$ is both submodular and symmetric.
\end{abstract}

\section{Introduction}\label{sec:intro}
We consider the \MLOPlong (\MLOP).
In this problem, we are given a ground set $N$ of cardinality $n$ and a non-negative set function
$f \colon 2^{N} \to \mathbb{R}_{\ge 0}$.
The goal is to find an ordering of the elements of $N$, i.e., a permutation $\pi \colon N \to \{1,\ldots,n\}$, that minimizes
the sum of the values of $f$ over all prefixes of $\pi$. 
Formally, the goal is to minimize:
\begin{align}
\sum_{i=1}^{n-1} f(N_{\pi,i}), \label{objective}
\end{align}
where $N_{\pi,i}$ denotes the first $i$ elements of $N$ with respect to $\pi$, i.e., $N_{\pi,i}\triangleq \{ u\in N\colon \pi(u)\leq i\}$.
For simplicity of presentation, we denote by $\val(\pi) \triangleq \sum_{i=1}^{n-1} f(N_{\pi,i})$ the objective value of an ordering $\pi$.

The problem was introduced in its full generality by Iwata, Tetali and Tripathi~\cite{iwata2012approximating}\footnote{The objective in~\cite{iwata2012approximating} is defined by summing over $i=0,\ldots,n$, which differs from the above definition by the ordering-independent term $f(\varnothing)+f(N)$, where $f(\varnothing)$ corresponds to $i=0$ and $f(N)$ corresponds to $i=n$.
Since $f$ is non-negative, removing $i=0$ and $i=n$ from the sum can only improve multiplicative approximation guarantees.
}, and has been studied for various types of set functions and special cases.
Two notable types of set functions are submodular and supermodular functions.
A set function $f\colon 2^N\to \mathbb{R}$ is submodular if: 
\begin{align}
    f(A)+f(B)\geq f(A\cup B) + f(A\cap B), ~~~~~~~~~~~\forall A, B\subseteq N,\nonumber
\end{align}
and $f$ is supermodular if $-f$ is submodular.

For a supermodular $f$, MLOP captures the classic \MSSClong problem.
In this problem, the ground set $N$ consists of $n$ sets that cover $m$ elements.
Given an ordering $\pi$ of $N$, every element is associated with a cover-time: the position of the first set in $\pi$ that covers the element.
The goal is to find an ordering $\pi$ that minimizes the sum of cover-times over all $m$ elements.
An interesting special case of \MSSClong is \MSVClong, where each element appears in two sets.

The study of both \MSSClong and \MSVClong was initiated by Feige, Lov\'{a}sz and Tetali~\cite{feige_MSSC}, and has admitted improved upper and lower bounds since its introduction, e.g., \cite{barenholz2006improved,bansal2021improved,iwata2012approximating,stankovic22} (the reader is referred to Section~\ref{sec:relatedwork} for a brief overview of related work).
For a general non-negative supermodular $f$, the approximability of \MLOP is resolved.
A $4$-approximation is given by~\cite{iwata2012approximating}, and a matching hardness is given by~\cite{feige_MSSC}: for every $\varepsilon >0$ achieving an approximation of $(4-\varepsilon)$ for \MSSClong is NP-hard.

For a submodular $f$, \MLOP captures several fundamental graph ordering problems.
Two prominent examples include \MLAlong and \MCIGlong, 
whose study dates back as early as the work of Leighton and Rao~\cite{leighton1999multicommodity} (\MCIGlong is also known as \MSlong, see, e.g., D{\'\i}az, Gibbons, Paterson and Toran~\cite{diaz1991minsumcut}, and as \MPlong, see, e.g., Lin and Yuan~\cite{Lin94}).
In both problems, we are given an undirected graph $G=(V,E)$ where
$(1)$ in the former problem, the goal is to find an ordering $\pi$ of $V$ that minimizes the total length of edges in the ordering: $\sum _{(u,v)\in E}| \pi(u)-\pi(v)| $; and
$(2)$ in the latter problem, the goal is to find an interval graph $ G'=(V,E')$ on the same nodes as $G$ that contains $G$ as a subgraph, i.e., $ E\subseteq E'$, and minimizes $|E'|$.

Clearly, \MLOP captures \MLAlong as a special case with a symmetric submodular $f$ that is the cut function of $G$, i.e., $f(S)$ equals the number of edges of $G$ that cross the cut $S$.
To see why \MLOP captures \MCIGlong, the classic characterization of Ramalingam and Pandu Rangan~\cite{RAMALINGAM88} is utilized: a graph $G=(V,E)$ is an interval graph if and only if there exists an ordering $\pi$ of $V$ such that for every edge $(u,v)\in E$, if $\pi(u)<\pi(v)$, then for every node $w$ satisfying $ \pi(u)<\pi(w)<\pi(v)$ there exists an edge connecting $w$ and $v$.
Thus, choosing a submodular, but not symmetric, $f$ that is the node cut function of $G$, i.e., $f(S)$ equals the number of nodes outside of $S$ that have neighbors in $S$, establishes that \MCIGlong is captured by \MLOP with a general submodular and non-symmetric $f$.
The study of both \MLAlong and \MCIGlong has attracted much attention in recent decades, e.g.,~\cite{arora2009expander,BL84,charikar2010,even2000divide,FHL08,feige_MLA_BEST,leighton1999multicommodity,rao2005new,ravi1991ordering,seymour1995packing}
(the reader is referred to Section~\ref{sec:relatedwork} for a brief overview of related work).

In contrast to \MLOP with a supermodular $f$, which is resolved, much less is known for \MLOP with a submodular $f$.
Results are known for only two special cases of a submodular $f$.
First, consider the special case that $f$ is submodular and monotone, i.e., $ f(A)\leq f(B)$, $\forall A\subseteq B\subseteq N$ (\MLOP with a monotone submodular $f$ captures additional applications, e.g., the \MLSClong problem; see, e.g., Hassin and Levin~\cite{hassin2005approximation}). 
For this special case,~\cite{iwata2012approximating} presented an approximation of $2(1-1/(n+1))$.
This bound was further improved for specific monotone submodular functions $f$ by Farhadi, Gupta, Sun, Tetali and Wigal~\cite{farhadi2024hardness}.
Second, consider the special case that $f$ is submodular and symmetric, i.e., $ f(A)=f(N\setminus A)$, $\forall A\subseteq N$.
For this special case, it is known that any ordering $\pi$ provides an approximation of $O(n)$~\cite{farhadi2024hardness} and an improved approximation of $O(\sqrt{n/\ln{n}})$ is given by Katzelnick and Schwartz~\cite{katzelnick2023simple}.
For this special case,~\cite{iwata2012approximating} show that for every $\varepsilon >0$ any algorithm that evaluates $f$ a polynomial number of times cannot achieve an approximation of $(2-\varepsilon)$.

Summarizing, while the approximability of \MLOP with a supermodular $f$ is resolved, much less is known when $f$ is submodular.
Specifically, for a general (not necessarily monotone or symmetric) submodular $f$, the best known hardness bound for \MLOP is $2$, and no approximation algorithm is known to the best of our knowledge.
The goal of this work is to close this gap.

\paragraph*{Value-Oracle Model.} As is typical in submodular optimization problems, we assume that the algorithm is not given $f$ explicitly, but can access $f$ using a value-oracle as follows: for every $S\subseteq N$, the algorithm can query the oracle and obtain $f(S)$.
Hence, the running time of the algorithm is measured by the number of value-oracle queries it performs (as well as the number of arithmetic operations).

\subsection{Our Result}\label{sec:results}
In this work, we resolve the approximability of \MLOP with a general non-negative submodular $f$.

First, we present an algorithm that achieves an approximation of $O(\sqrt{n/\ln{n}})$ for \MLOP, as is summarized by the following theorem.
\begin{theorem}\label{thm:upper_bound}
There exists a polynomial-time algorithm that, given a non-negative submodular set function
$f$ over a ground set $N$ of size $n$,
outputs an ordering $\pi$ of $N$ such that:
\[
\val (\pi)
\le
O\!\left(\sqrt{n/\ln{n}}\right)\cdot \OPT.
\]
Here, $\OPT$ denotes the value of an optimal solution to the given instance of the \MLOPlong.
\end{theorem}
Previously, an approximation of $ O(\sqrt{n/\ln{n}})$, given by Katzelnick and Schwartz \cite{katzelnick2023simple}, was known only in the special case where $f$ is not only submodular but also symmetric.
To the best of our knowledge, no algorithm for \MLOP with a general non-negative submodular $f$ was previously known.
It should be noted that our algorithm differs from the algorithm of \cite{katzelnick2023simple}, as the latter algorithm and its analysis are crucially based on the symmetry of $f$.
Thus, a new approach is required (see Section~\ref{sec:techniques} for an overview of our techniques).

Second, we present a hardness result for \MLOP that matches the guarantee of Theorem~\ref{thm:upper_bound}, as is summarized by the following theorem.
\begin{theorem}\label{thm:lower_bound}
No (possibly randomized) algorithm that makes a polynomial number in $n$ of value-oracle queries
can achieve an $o(\sqrt{n/\ln n})$-approximation for the \MLOPlong
with a non-negative submodular objective $f$ over a ground set $N$ of size $n$.
\end{theorem}
There are two things to note regarding Theorem~\ref{thm:lower_bound}.
First, the hardness result in Theorem~\ref{thm:lower_bound} improves the previously best known hardness of $2$, which was given by Iwata, Tetali and Tripathi~\cite{iwata2012approximating}, and asymptotically matches the $O(\sqrt{n/\ln{n}})$-approximation guarantee of Theorem~\ref{thm:upper_bound}.
Thus, the approximability of \MLOP with a submodular $f$ is resolved.
Second, it should be noted that the hardness result of Theorem~\ref{thm:lower_bound}, as is common for submodular optimization problems in the value-oracle model, is information-theoretic and does not rely on any computational complexity assumptions.

\subsection{Our Techniques}\label{sec:techniques}
\paragraph*{Algorithmic Approach.}
A classic framework of recursive balanced cutting, which dates back to Bhatt and Leighton~\cite{BL84}, can be successfully applied to many graph ordering problems.
In a nutshell, given a subset of vertices $A$, the algorithm:
$(1)$ computes a suitable balanced cut $ S\subseteq A$ in the subgraph $G[A]$ induced by $A$;
$(2)$ places the vertices of $S$ before the vertices of $A\setminus S$; and
$(3)$ recurses on $G[S]$ and $G[A\setminus S]$, the subgraphs induced by $S$ and $A\setminus S$, respectively, to determine the ordering within each part.
For example, a suitable balanced cut problem in step $(1)$ above is the classic Minimum $b$-Balanced Cut problem for \MLAlong and the classic Node $\rho$-Separator problem for \MCIGlong.

Roughly speaking, if one has an $\alpha$-approximation for the relevant balanced cut problem, it is known that the above framework provides an approximation of $O(\alpha\ln{n})$ for many ordering problems, including both \MLAlong and \MCIGlong.
Intuitively, two essential structural observations enable one to prove this approximation.
To simplify the presentation of these observations and keep the discussion as informal as possible, in what follows we focus on \MLAlong (a similar discussion involving nodes, as opposed to edges, also applies to \MCIGlong).
First, each recursive call operates on $G[A]$, i.e., an induced subgraph whose nodes are exactly $A$ and whose edges consist only of edges with both endpoints in $A$.
Consequently, once a cut is made, edges crossing between the two sides $S$ and $ A\setminus S$ are permanently removed.
Second, the objective function changes across recursive calls, as it becomes the cut function of the relevant induced subgraph.
Together, these two observations ensure that each edge is charged only once, resulting in the claimed approximation.

Unfortunately, when considering \MLOP, these observations fail and the framework breaks down.
The main reason is that there is no clear notion of an induced subproblem, since the input consists only of $N$ and the set function $f$.
Therefore, there is no obvious underlying structure that can be deleted as the recursion proceeds.
Moreover, it can be easily shown that trivially applying the recursive balanced cutting framework by recursively computing balanced cuts on each subset $A$ with respect to the original $f$ fails.

Katzelnick and Schwartz~\cite{katzelnick2023simple} address this difficulty for the special case where $f$ is not only submodular but also symmetric.
Specifically, when recursing on $A$,~\cite{katzelnick2023simple} find a balanced cut $S\subseteq A$ of $A$ while using a different objective function.
Intuitively, $f$ is modified in a manner that depends on $A$, allowing one to mimic the induced subgraph behavior and the proof of the recursive balanced cutting framework for \MLAlong.
However, this solution relies crucially on the symmetry of $f$, and without it the algorithm of~\cite{katzelnick2023simple} is not even properly defined.
Thus, a different approach is required when considering \MLOP with a general, and not necessarily symmetric, submodular $f$.

Our approach is remarkably simple: when recursing on $A$, we find a {\bf{global cut}} $S\subseteq N$ using the {\bf{original set function}} $f$, where the only requirement is that $S$ is balanced with respect to $A$, i.e., $ b\cdot \lvert A\rvert\leq \lvert S\cap A\rvert\leq (1-b)\cdot \lvert A\rvert$ for some constant $ 0<b< \nicefrac{1}{2}$.
Then, our algorithm places $S\cap A$ before $A\setminus (S \cap A)$ and recurses on both.
Hence, instead of restricting the cut $S$ to be a subset of $A$ and changing the function $f$, as is done when applying the recursive balanced cutting framework to \MLAlong and in~\cite{katzelnick2023simple} for \MLOP with a symmetric submodular $f$, we:
$(1)$ compute the cut $S$ over the entire ground set $N$; and
$(2)$ use the original $f$ throughout the recursion regardless of the current set $A$.

At first glance, this approach appears to be inherently wasteful.
To exemplify this intuition, focus on \MLAlong.
For this problem, our approach reduces to the following when recursing on $A\subseteq V$:
we compute an $\alpha$-approximate global cut $S\subseteq V$ that is balanced with respect to $A$ and minimizes the total number of edges crossing $S$ in the entire graph $G$.
Therefore, the cost of $S$ is determined also by edges outside of $G[A]$, e.g., edges with both endpoints outside of $A$ that cross $S$.
Thus, the induced subgraph $G[A]$ is ignored when determining the cost of $S\subseteq V$.
Nonetheless, a surprising consequence of our work is that the above algorithm, which seems wasteful, also provides an approximation of $ O(\alpha \ln{n})$ for \MLAlong.

To analyze our algorithm, we first observe that there is a one-to-one correspondence between (non-base) recursive calls and prefixes in the final ordering of the algorithm.
We then prove that each prefix admits a recursive formula that consists of only union and intersection operations over the approximate balanced global cuts computed along the chain of recursive calls starting with the initial call on $N$ and ending with the recursive call corresponding to the prefix.
The crucial observation is that each such balanced global cut appears exactly once in this formula, i.e., the formula is read-once.
This property allows us to exploit the submodularity and non-negativity of $f$ to upper bound the value of the final ordering by a sum over all (non-base) recursive calls, where each (non-base) recursive call on $A$ contributes to this sum a value of $\lvert A\rvert\cdot f(S)$ (here $S\subseteq N$ is the $\alpha$-approximate balanced global cut computed by the algorithm in the recursive call on $A$).
This suffices to prove the desired upper bound on the value of the final ordering of the algorithm.

\paragraph*{Hardness Approach.}
Our hardness is proved using an indistinguishability argument in the value-oracle
model, a classic approach in submodular optimization, see, e.g.,~\cite{NF78,svitkina2011submodular,Feige_Hardness}.
In this approach, one typically constructs two submodular functions: a deterministic function
$f_1$ and a randomized function $f_2$. 
The functions are required to be indistinguishable to any algorithm making a polynomial number of value-oracle queries, yet the values of their optimal solutions differ substantially.
Thus, any algorithm making a polynomial number of value-oracle queries achieving an approximation better than the ratio of the optimal solutions would therefore be able to distinguish
between $f_1$ and $f_2$, contradicting their indistinguishability.
Hence, the optima ratio translates into information-theoretic hardness of the same value.

Existing hardness constructions following this framework, including hardness constructions based on the symmetry gap machinery of Vondr\'{a}k~\cite{vondrak2013symmetry_SymmetryGap}, typically hide a random set.
Such constructions are well suited to various minimization and maximization submodular optimization problems whose
solution is characterized by a single set.
Unfortunately, such an approach is not sufficient for \MLOP, whose solution consists of an ordering of the entire ground set $N$.
Consequently, these techniques do not yield a strong hardness result for \MLOP (the current best known hardness for \MLOP with a submodular $f$ equals only $2$~\cite{iwata2012approximating} and it is obtained by hiding a single random set).

Our construction differs in that the randomized function $f_2$ hides a uniform random
ordering rather than just a set. 
Consequently, this results in incorporating all
suffixes of the random ordering into the definition of $f_2$.
While this is essential for obtaining the desired gap for \MLOP, it also introduces new challenges.
The main challenge is that despite the many suffixes incorporated into the definition of $f_2$, it must remain
non-negative, submodular, and indistinguishable from $f_1$.
To guarantee that all of the above are achieved simultaneously for all possible random orderings, we leverage an observation of Lov\'{a}sz~\cite{lovasz1983submodular} that provides a sufficient condition when the minimum of two submodular functions is also submodular.

\subsection{Related Work}\label{sec:relatedwork}
\MSSClong was first considered by Feige, Lov\'{a}sz and Tetali~\cite{feige_MSSC}, who presented an approximation of $4$ for it and showed that this approximation is tight by proving that, for every $ \varepsilon>0$, achieving an approximation of $(4-\varepsilon)$ is NP-hard.
\MSVClong was also introduced in~\cite{feige_MSSC}, who presented an approximation of $2$, and this approximation was subsequently improved to $\approx 1.999946$ by Barenholz, Feige and Peleg~\cite{barenholz2006improved} and to $\nicefrac{16}{9}$ by Bansal, Batra, Farhadi and Tetali~\cite{bansal2021improved}, which is the current best known approximation.
Focusing on lower bounds for \MSVClong, it was proved in~\cite{feige_MSSC} that there exists a constant $\rho>1$ for which it is NP-hard to obtain an approximation better than $\rho$, and Stankovi\'{c}~\cite{stankovic22} proved that an approximation better than $ 1.014$ is not possible assuming the Unique Games Conjecture.

Both \MLAlong and \MCIGlong admit a rich history,
and as previously mentioned, the recursive balanced cutting framework provides an approximation of $ O(\alpha\ln{n})$ for both problems (where $\alpha$ is the approximation achieved for the relevant balanced cut problem).
Specifically, for \MLAlong, Leighton and Rao~\cite{leighton1999multicommodity} present an approximation algorithm with $ \alpha=O(\ln{n})$ which was subsequently improved by Arora, Rao and Vazirani~\cite{arora2009expander} to $ O(\sqrt{\ln{n}})$, yielding overall approximations of $ O(\ln^2{n})$ and $ O(\ln^{\nicefrac{3}{2}}{n})$, respectively.
Similarly, for \MCIGlong, Ravi, Agrawal and Klein~\cite{ravi1991ordering} and Feige, Hajiaghayi and Lee~\cite{FHL08} provide overall approximations of $ O(\ln^2{n})$ and $ O(\ln^{\nicefrac{3}{2}}{n})$, respectively.

Going beyond recursive balanced cutting, both problems admit improved results: Even, Naor, Rao and Schieber~\cite{even2000divide} provide an overall approximation of $O(\ln{n}\ln{\ln{n}})$ using the cutting scheme of Seymour~\cite{seymour1995packing};
Rao and Richa~\cite{rao2005new} improve the approximation factor to $O(\ln{n})$; and finally, by building on the work of~\cite{arora2009expander}, both Charikar, Hajiaghayi, Karloff and Rao~\cite{charikar2010} and Feige and Lee~\cite{feige_MLA_BEST} improve the approximation factor to $ O(\sqrt{\ln{n}}\ln{\ln{n}})$.

A substantial line of work has established hardness results for submodular optimization problems in the value-oracle model using indistinguishability-based constructions.
This approach dates back to Nemhauser and Fisher~\cite{NF78} from the late $70$'s, and subsequent works have applied this approach to a wide range of submodular optimization problems, e.g.,~\cite{Feige_Hardness,goemans2009approximating,svitkina2011submodular,iwata_hardness_example,goel2009approximability,santiago2019multivariate,mirrokni2008tight}.

Indistinguishability-based constructions were subsequently captured by the symmetry gap machinery, introduced by Vondrák~\cite{vondrak2013symmetry_SymmetryGap} to prove information-theoretic hardness bounds for submodular maximization problems.
Since its introduction, the symmetry gap machinery has been used to obtain multiple value-oracle hardness results for a variety of problems, see, e.g.,~\cite{gharan2011submodular,azar2011submodular,filmus2025separating}.
Subsequently, Ene, Vondrák, and Wu~\cite{ene2013local} showed that the symmetry gap machinery can also be used to derive value-oracle hardness results for submodular minimization problems.
Moreover, Dobzinski and Vondrák~\cite{dobzinski2012query} showed how symmetry gap based value-oracle hardness results can be translated into computational complexity inapproximability results, assuming $ \text{NP}\neq \text{RP}$, when $f$ is given explicitly.

\paragraph*{Paper Organization.}
Section~\ref{sec:preliminaries} contains several required preliminary results.
Section~\ref{sec:algorithm} focuses on our algorithm, whereas Section~\ref{sec:hardness} contains the hardness result.

\section{Preliminaries}\label{sec:preliminaries}
Throughout the paper $\OPT$ denotes the value of an optimal solution for \MLOP.

\paragraph*{Submodular $b$-Balanced Cut.}
Our algorithm requires an approximation algorithm to a weighted variant of the \SBClong problem.
In the \SBClong problem, we are given a ground set $V$ of size $n$ that is equipped with non-negative weights $ w\colon V\rightarrow \mathbb{R}_{\geq 0}$, a non-negative submodular function $ f\colon 2^V \rightarrow \mathbb{R}_{\geq 0}$, and a balancing parameter $ b<\nicefrac{1}{2}$.
A subset $ A\subseteq V$ is a $b$-balanced cut if $ b\cdot w(V)\leq w(A)\leq (1-b)\cdot w(V)$, where $w(S)\triangleq \sum _{u\in S}w(u)$ for every $ S\subseteq V$.
The goal is to find a $b$-balanced cut $A\subseteq V$ minimizing $f(A)$:
\begin{align}
  \min \{ f(A)\colon A\subseteq V,\; b\cdot w(V)\leq w(A)\leq (1-b)\cdot w(V) \} . \nonumber
\end{align}
The following theorem was given by Svitkina and Fleischer~\cite{svitkina2011submodular} and it provides a bicriteria approximation for \SBClong.
\begin{theorem}[Svitkina and Fleischer~\cite{svitkina2011submodular}]
\label{thm:balanced_cut_algorithm}
There exists a polynomial time algorithm that given a ground set $V$ of size $n$ equipped with non-negative weights $w\colon V\to \mathbb{R}_{\geq 0}$, a non-negative submodular function $f\colon 2^V\rightarrow \mathbb{R}_{\geq 0}$, a balancing parameter $ 0<b\leq \nicefrac{1}{2}$, and a bicriteria parameter $ 0<b'<b$, finds a $\nicefrac{b'}{2}$-balanced cut $S\subseteq V$ satisfying:
\begin{align}
    f(S)\leq O\left( \frac{\sqrt{k/\ln{k}}}{b-b'}\right)\cdot \min \{ f(A)\colon A\subseteq V, \; b\cdot w(V)\leq w(A)\leq (1-b)\cdot w(V) \}, \nonumber
\end{align}
where $k=\lvert\{ u\in V\colon w(u)>0\}\rvert$ is the size of the support of $w$.
\end{theorem}
There are two things to note regarding Theorem~\ref{thm:balanced_cut_algorithm}:
$(1)$ the approximation guarantee proved in~\cite{svitkina2011submodular} is only $ O(\sqrt{n/\ln{n}}/(b-b'))$ and not $ O(\sqrt{k/\ln{k}}/(b-b'))$, however essentially the same algorithm and proof as in~\cite{svitkina2011submodular} provide the improved guarantee of Theorem~\ref{thm:balanced_cut_algorithm}; and
$(2)$ the algorithm of~\cite{svitkina2011submodular} is randomized and fails with a probability that is arbitrarily small, e.g., exponentially small in $n$, however to simplify presentation throughout the paper we explicitly omit any discussion on success probability and condition on a successful execution of the algorithm of Theorem~\ref{thm:balanced_cut_algorithm}.
A formal discussion regarding both points above appears in Appendix~\ref{app:ssc_probability}.

\paragraph*{Submodularity Preserving Operations.}
Our tight information-theoretic hardness result requires the following lemma given by Lov\'{a}sz~\cite{lovasz1983submodular}.
\begin{lemma}[Lov\'{a}sz~\cite{lovasz1983submodular}]\label{lem:lovasz}
Let $f$ and $g$ be two submodular set functions over the same ground set such that $f-g$ is either monotone non-decreasing or monotone non-increasing.
Then, the set function $\min\{ f,g\}$ is also submodular.
\end{lemma}

\paragraph*{Indistinguishability.}
We say that an algorithm distinguishes two functions $f_1$ and $f_2$ if its output
differs when given value oracle access to $f_1$ versus value oracle access to $f_2$. 
The following lemma, due to Svitkina and Fleischer~\cite{svitkina2011submodular},
is used to establish an approximation lower bound in the value-oracle model;
it requires two conditions: the first is sufficient to guarantee oracle
indistinguishability between the functions, and the second requires a
multiplicative gap between their optimal objective values.
\begin{lemma}[Svitkina and Fleischer~\cite{svitkina2011submodular}]
\label{lem:general_Indistinguishability_def}
Let $\mathcal{P}$ be a minimization problem defined on set functions over a ground
set $N$ of size $n$.
If there exist a set function $f_1$ on $N$ and a distribution $\mathcal{D}$ over set
functions on $N$ such that
$(1)$ for every set $S \subseteq N$,
$\Pr_{f_2 \sim \mathcal{D}}\!\left[f_1(S) \neq f_2(S)\right] \le n^{-\omega(1)}$;
and $(2)$ for every $f_2$ drawn from $\mathcal{D}$,
$\OPT_1 \ge \gamma \cdot \OPT_2$ for some $\gamma \ge 1$, where $\OPT_1$ and
$\OPT_2$ denote the optimal values of $\mathcal{P}$ on inputs $f_1$ and $f_2$,
respectively, then no (possibly randomized) algorithm making a polynomial number
of value-oracle queries can approximate $\mathcal{P}$ within a factor of $o(\gamma)$.
\end{lemma}

\section{\texorpdfstring{$O(\sqrt{n/\ln{n}})$-Approximation Algorithm}{O(sqrt(n/log n))-Approximation Algorithm}}\label{sec:algorithm}
As mentioned in Section~\ref{sec:techniques}, at a high level, our algorithm is inspired by the classic recursive balanced cutting framework, which has been extensively used in the literature for various ordering problems, see, e.g., Leighton and Rao~\cite{leighton1999multicommodity} and the references therein.
However, two key differences distinguish our approach from the recursive balanced cutting framework:
$(1)$ the balanced cut $S$  is computed over the entire ground set $N$ regardless of the current $A$, where $S$ only needs to be balanced with respect to $A$; and
$(2)$ the original $f$ is used regardless of $A$.

Specifically, given $A\subseteq N$, let $\ind_A$ denote the weight function that assigns a weight of $1$ to elements of $A$ and a weight of $0$ to all other elements, and the algorithm splits $A$ as follows:
$(1)$ it computes an approximate \SBClong $S\subseteq N$ of the entire ground set $N$ with element weights $ w=\ind_A$ and the given $f$ (as well as constants $b$ and $b'$ as in Theorem~\ref{thm:balanced_cut_algorithm});
$(2)$ it places the elements of $S \cap A$ before the elements of $ A\setminus (S \cap A)$, i.e., the elements of $S \cap A$ are placed to the left of the elements of $A\setminus (S \cap A)$; and
$(3)$ it recurses on $S \cap A$ and $ A\setminus (S \cap A)$.
The initial recursive invocation of the algorithm is performed with $A=N$.

To simplify presentation, for two orderings $\pi_1$ of $ A\subseteq N$ and $\pi_2$ of $B\subseteq N$, where $ A\cap B=\varnothing$, denote by $\pi_1 \circ \pi_2$ the ordering of $ A\cup B$ obtained by concatenating $\pi_2$ after $\pi_1$, i.e., all elements of $A$ are placed before (to the left of) all elements of $B$ and the order between elements of $A$ is determined by $\pi_1$ and the order between elements of $B$ is determined by $ \pi_2$. 

\begin{algorithm}[t]
\caption{Recursive Balanced Global Cutting}
\label{alg:MLOP_Alg}
\begin{algorithmic}[1]
\Require $A \subseteq N$
\Ensure ordering $\pi_A$ of $A$

\If{$\lvert A\rvert = 1$}
    \State \Return the unique ordering of $A$
\Else
    \State $w \gets \ind_A$
    \State execute the algorithm of Theorem~\ref{thm:balanced_cut_algorithm} on $ (N,f,w)$ with $ b=\nicefrac{1}{4}$ and $ b'=\nicefrac{1}{8}$ to obtain $S\subseteq N$ \label{line:calling_SBC}
    \State $T \gets S \cap A$
    \State recursively invoke Algorithm~\ref{alg:MLOP_Alg} on $T$ to obtain $\pi_T$
    \State recursively invoke Algorithm~\ref{alg:MLOP_Alg} on $A\setminus T$ to obtain $\pi_{A\setminus T}$
    \State $\pi_A \gets \pi_T \circ \pi_{A\setminus T}$
    \State \Return $\pi_A$
\EndIf

\end{algorithmic}
\end{algorithm}

\paragraph*{Analysis Notations.} We introduce the following notations to simplify the analysis.
First, denote by $\mathcal{A}$ the collection of sets $A\subseteq N$ for which Algorithm~\ref{alg:MLOP_Alg} invokes the approximation algorithm of Theorem~\ref{thm:balanced_cut_algorithm} for \SBClong, i.e., the non-base recursive calls.
Second, for any $ A\in \mathcal{A}$, denote by $h(A)\subseteq N$ the set of
elements Algorithm~\ref{alg:MLOP_Alg} places strictly before $A$ in the final ordering, i.e., elements placed to the left of $A$ in the final ordering.
Third and last, for every $ A\in \mathcal{A}$, denote by $\Alg(A)\subseteq N$ the set of elements returned by the \SBClong algorithm of Theorem~\ref{thm:balanced_cut_algorithm} when executed on $ (N,f,w)$ with $w=\ind_A$, $ b=\nicefrac{1}{4}$, and $ b'=\nicefrac{1}{8}$ (line~\ref{line:calling_SBC} of Algorithm~\ref{alg:MLOP_Alg}).

Following the above notations, one can view the execution of Algorithm~\ref{alg:MLOP_Alg} as inducing a decomposition tree whose internal (non-leaf) nodes are $\mathcal{A}$, which has $n$ leaves each corresponding to a different element of $N$, and whose root is $N$.
Given an internal node $A\in \mathcal{A}$, $A$ has two children:
the left child is $S \cap A$ and the right child is $A\setminus (S\cap A)$ (recall that $S\subseteq N$ is the approximate \SBClong found by Algorithm~\ref{alg:MLOP_Alg} in line~\ref{line:calling_SBC}).
Moreover, for every $i\geq 0$, denote by $ \mathcal{A}_i$ the collection of all $ A\in \mathcal{A}$ that are at depth $i$ in the decomposition tree, and let $ D\triangleq \max \{ i\colon \mathcal{A}_i\neq \varnothing \}$ denote the maximum depth of an internal (non-leaf) node of the decomposition tree.

\paragraph*{Analysis Overview.}
The goal in analyzing Algorithm~\ref{alg:MLOP_Alg} is to prove Theorem~\ref{thm:upper_bound}.
To this end, two main lemmas are required.

The first lemma provides an upper bound on $ \val(\pi)$, where $ \pi$ is the final ordering Algorithm~\ref{alg:MLOP_Alg} outputs, using only the values of the approximate solutions to all instances of \SBClong computed throughout the entire recursive execution of  Algorithm~\ref{alg:MLOP_Alg}.
It should be emphasized that this upper bound does not contain values of sets besides the approximate balanced cut solutions, making it easily relatable to $\OPT$.
This is summarized by the following lemma. 
\begin{lemma}\label{lem:alg_val_upper_bound}
Let $ \pi$ be the ordering Algorithm~\ref{alg:MLOP_Alg} outputs. Then,
\[
\val(\pi)
\le
\sum_{A \in \mathcal{A}} \lvert A\rvert\cdot f\bigl(\Alg(A)\bigr).
\]
\end{lemma}
We note that Lemma~\ref{lem:alg_val_upper_bound} captures the bulk of the proof of Theorem~\ref{thm:upper_bound}.
Moreover, its correctness heavily relies on proving that every prefix $ N_{\pi,i}$ corresponds to a different and unique $A\in \mathcal{A}$ and that $N_{\pi,i}$ can be written exactly as a read-once formula using only approximate balanced cut solutions of nodes in the unique path in the decomposition tree between $A$ and the root $N$.

The second structural lemma provides a lower bound on $\OPT$ using optimal values of disjoint instances of \SBClong, and is independent of Algorithm~\ref{alg:MLOP_Alg}.
This is summarized by the following lemma.
\begin{lemma}\label{lem:structural-decomposition}
Let $X_1,\dots,X_\ell \subseteq N$ be pairwise disjoint subsets, and assume that $\lvert X_i\rvert\geq 2$ for all $i\in [\ell]$.
For each $i \in [\ell]$, let $S_i^\star \subseteq N$ denote an optimal solution to the \SBClong instance $(N,f,w)$ with $w = \mathbf{1}_{X_i}$ and $ b=\nicefrac{1}{4}$.
Then,
\[
\sum_{i =1}^{\ell} \lvert X_i\rvert \cdot f(S_i^\star) \leq c\cdot \OPT,
\]
for some absolute constant $c>0$.
\end{lemma}

Theorem~\ref{thm:upper_bound} follows from combining Lemmas~\ref{lem:alg_val_upper_bound} and~\ref{lem:structural-decomposition}.
\begin{proof}[Proof of Theorem~\ref{thm:upper_bound}]
We show that Algorithm~\ref{alg:MLOP_Alg} satisfies the guarantee of Theorem~\ref{thm:upper_bound}.
Let $\pi$ denote the final ordering Algorithm~\ref{alg:MLOP_Alg} outputs.

Lemma~\ref{lem:alg_val_upper_bound} provides the following upper bound on $ \val(\pi)$:
\begin{align}
\val(\pi)
\le
\sum_{A \in \mathcal{A}} \lvert A\rvert\cdot f\bigl(\Alg(A)\bigr).  \label{UB-ineq1}
\end{align}
Since $\mathcal{A} = \biguplus_{i=0}^{D} \mathcal{A}_i$, where $ \biguplus$ denotes disjoint union, the right-hand side of~\eqref{UB-ineq1} can be rewritten by layers of the decomposition tree as follows:
\begin{align}
\sum_{A\in\mathcal{A}}\lvert A\rvert\cdot f\bigl(\Alg(A)\bigr)
= \sum_{i=0}^{D} \sum_{A\in\mathcal{A}_i} \lvert A\rvert\cdot f\bigl(\Alg(A)\bigr). \label{UB-eq1}
\end{align}

The contribution of each layer $ \mathcal{A}_i$ to the right-hand side of~\eqref{UB-eq1} can be upper bounded by:
\begin{align}
 \sum_{A \in \mathcal{A}_i} \lvert A\rvert\cdot f\bigl(\Alg(A)\bigr)
\le
O\!\left(\sqrt{(\beta^i\,n)/(\ln(\beta^i n))}\right)
\cdot \OPT, \label{UB-layer}   
\end{align}
where $ \beta=\nicefrac{15}{16}$.
To prove~\eqref{UB-layer}
assume that $\mathcal{A}_i = \{A_1,\dots,A_\ell\}$.
For each $j \in [\ell]$, let $S_j^\star$ denote an optimal solution to \SBClong for the instance $(N,f,w)$ with $w = \ind_{A_j}$ and $b=\nicefrac{1}{4}$.
Since $\Alg(A_j)$ is obtained by applying Theorem~\ref{thm:balanced_cut_algorithm} on this instance with $b'=\nicefrac{1}{8}$:
\begin{align}
f\bigl(\Alg(A_j)\bigr)
\le
O\!\left(\sqrt{\lvert A_j\rvert/\ln \lvert A_j\rvert}\right) \cdot f(S_j^\star). \label{UB-layer2}
\end{align}
The function $x/\ln x$ is increasing for all $x \ge e$, and
$\lvert A_j\rvert \le \beta^i n$ for every $A_j \in \mathcal{A}_i$ (recall the balancing guarantee of Theorem~\ref{thm:balanced_cut_algorithm} and the choice of $b'=\nicefrac{1}{8}$).
Hence,the right-hand side of~\eqref{UB-layer2} can be further upper bounded as follows:
\begin{align}
f\bigl(\Alg(A_j)\bigr)
\le O\!\left(\sqrt{(\beta^i\,n)/\ln(\beta^i n)}\right)\cdot f(S_j^\star) .\label{UB-layer3}
\end{align}
Summing over~\eqref{UB-layer3} for all $ A\in \mathcal{A}_i$ and noting that $ A_1,\ldots,A_{\ell}$ are pairwise disjoint by construction, Lemma~\ref{lem:structural-decomposition} implies that:
\begin{align*}
    \sum_{A \in \mathcal{A}_i} \lvert A\rvert\cdot f\bigl(\Alg(A)\bigr)
&\le O\!\left(\sqrt{(\beta^i\,n)/\ln(\beta^i n)}\right)\sum_{j=1}^\ell \lvert A_j\rvert\cdot f(S_j^\star) 
\leq O\!\left(\sqrt{(\beta^i\,n)/\ln(\beta^i n)}\right) \cdot \OPT.
\end{align*}
Thus,~\eqref{UB-layer} holds.

Therefore,
\begin{align}
\val(\pi)
&\le
\sum_{i=0}^{D}
O\!\left(
\sqrt{(\beta^i\, n)/\ln\bigl(\beta^i\, n\bigr)}
\right)\cdot \OPT 
=
O\!\left(\sqrt{n/\ln n}\right)\cdot \OPT, \label{UB-eq2}
\end{align}
where: $(1)$ the inequality in~\eqref{UB-eq2}
follows from combining~\eqref{UB-ineq1} and~\eqref{UB-eq1} and applying~\eqref{UB-layer}  
to each layer $\mathcal{A}_i$;
and $(2)$ the equality in~\eqref{UB-eq2} follows since by definition, layer $\mathcal{A}_D$ contains a subset $A$ with $\lvert A\rvert\geq 2$ and $ \lvert A\rvert\leq \beta^Dn$, thus implying that $ D\leq \ln{(n/2)}/\ln{(1/\beta)}$.
Hence, standard arithmetic (see Lemma~\ref{lem:final_sum_upper_bound} in Appendix~\ref{app:results}) suffices to conclude the proof.
\end{proof}

\paragraph*{Proving Lemma~\ref{lem:alg_val_upper_bound}.}
Lemma~\ref{lem:alg_val_upper_bound} is the crux of the proof of Theorem~\ref{thm:upper_bound}, and its proof requires two preliminary steps.

First, the value of the solution is related to the decomposition tree.
Specifically, the value of the final ordering $\pi$ produced by Algorithm~\ref{alg:MLOP_Alg} is expressed  as a sum of terms associated with the non-base recursive invocations as follows:
\begin{align}
\val(\pi)
=
\sum_{A \in \mathcal{A}}
f\!\left( \left(\Alg(A) \cap A\right) \cup \hS(A) \right). \label{objective-decomposition}    
\end{align}
Recall that for any $ A\in \mathcal{A}$, $h(A)\subseteq N$ denotes the set of
elements that Algorithm~\ref{alg:MLOP_Alg} places strictly before $A$ in the final ordering $\pi$, i.e., elements placed to the left of $A$ in the final ordering $\pi$.

To see why~\eqref{objective-decomposition} is true, consider a non-base recursive call on a subset $A \in \mathcal{A}$.
Algorithm~\ref{alg:MLOP_Alg} invokes Theorem~\ref{thm:balanced_cut_algorithm} and obtains a subset $\Alg(A)\subseteq N$, and defines $T \leftarrow \Alg(A) \cap A$.
At this point, the algorithm commits to placing all elements of $T$ strictly
before, i.e., to the left of, all elements of $A \setminus T$ in the final ordering.
Subsequent recursive calls only determine the relative order within $T$ and
within $A \setminus T$, but cannot violate the above ordering constraint.
Thus, as a consequence of this commitment, in the final ordering $\pi$ there exists
a necessarily unique index $t$ satisfying: $ N_{\pi,t} = \hS(A) \cup T$.
Hence, the recursive call on $A$ corresponds to a unique prefix $N_{\pi,t}$.
Moreover, distinct recursive calls correspond to distinct prefixes, and the associated
term $f(N_{\pi,t})$ in the objective can be written as
$f\!\left((\Alg(A)\cap A)\cup \hS(A)\right)$.
Summing the corresponding terms over all $A \in \mathcal{A}$ yields~\eqref{objective-decomposition}.

Second, an additional recursive sequence of subsets associated with the execution of the algorithm is needed.
To simplify presentation, a mild abuse of notation is introduced, as $A$ is used to denote both an internal (non-leaf) node in the decomposition tree and the set associated with the corresponding recursive call.
For a node $A\in \mathcal{A}$, let $A^{(i)}$ denote the $i$\textsuperscript{th} ancestor of $A$ in the decomposition tree, where
$A^{(0)} \triangleq A$, $A^{(1)}$ is the immediate parent of $A$, and so on, until reaching the root.
Let $d(A)$ denote the depth of $A$ in the decomposition tree, i.e., the unique integer such that $A^{(d(A))}=N$ (recall that $N$ is the root of the decomposition tree).

For a node $A\in \mathcal{A}$ and a set $Z \subseteq N$, consider the following sequence of subsets whose definition is recursive.
The base case is $ E_0 \triangleq Z$ and for $i = 1,2,\dots,d(A)$ the recursion is:
\[
E_i \triangleq
\begin{cases}
\Alg\!\bigl(A^{(i)}\bigr)\ \cap\ E_{i-1},
& \text{if } A^{(i-1)} \text{ is a left child of } A^{(i)},\\[4pt]
\Alg\!\bigl(A^{(i)}\bigr)\ \cup\ E_{i-1},
& \text{if } A^{(i-1)} \text{ is a right child of } A^{(i)}.
\end{cases}
\]
The subset of interest, denoted by $ T_A(Z)$, is the last subset in the above sequence: $T_A(Z) \triangleq  E_{d(A)}$.

The following lemma captures the key structural property of the construction of $T_A(Z)$.
It proves that when instantiated with $Z=\Alg(A)$, the set $T_A(Z)$ coincides exactly with the contribution of the prefix that corresponds to the recursive call on $A$ to $\val(\pi)$~\eqref{objective-decomposition}, i.e., $ T_A(Z)=(\Alg(A)\cap A)\cup h(A)$.
This identity is crucial, as $T_A(Z)$ admits by definition a simple read-once formula, implying that $(\Alg(A)\cap A)\cup h(A) $ admits the same read-once formula.

\begin{lemma}
\label{clm:TA_general_identity}
For every $A \in \mathcal{A}$ and every $Z \subseteq N$: $ T_A(Z) = (Z \cap A)\ \cup\ \hS(A)$.
\end{lemma}
\begin{proof}
We prove the lemma by induction on the depth $d(A)$ of node $A$ in the decomposition tree.

First, consider the base case, i.e., $A=N$ is the root of the decomposition tree with $d(A)=0$.
By definition $T_N(Z)=E_0=Z$.
Moreover, $h(N)=\varnothing$ implying that $ (Z\cap N)\cup h(N)=Z$.
Thus, the lemma holds for the base case.

Second, assume the lemma holds for all nodes in the decomposition tree of depth at most $i$ and let $A$ be a node of depth $i+1$.
Denote the parent of $A$ by $P$.
Consider two cases depending on whether $A$ is a left or right child of $P$ in the decomposition tree.

Assume $A$ is a left child of $P$, i.e., $ A=\Alg(P)\cap P$ by the definition of Algorithm~\ref{alg:MLOP_Alg}.
Observe that in this case the definitions of $T_A(\cdot)$ and $ T_P(\cdot)$ imply that:
\begin{align}
  T_A(Z)=T_P(\Alg(P)\cap Z).  \label{UB-eq4}
\end{align}
Applying the induction hypothesis to $P$ yields that: $ T_P(\Alg(P)\cap Z)=((\Alg(P)\cap Z)\cap P)\cup h(P)$.
Thus,
\begin{align}
   T_P(\Alg(P)\cap Z) =(A\cap Z)\cup h(P)=(A\cap Z)\cup h(A),\label{UB-eq5}
\end{align}
where:
$(1)$ the first equality follows since $ A=\Alg(P)\cap P$; and
$(2)$ the second equality follows since $ h(A)=h(P)$ (recalling that $A$ is the left child of $P$).
Combining~\eqref{UB-eq4} and~\eqref{UB-eq5} concludes the proof of the inductive step in the case where $A$ is a left child.

Assume $A$ is a right child of $P$, i.e., $A=P\setminus (\Alg(P)\cap P)$ by the definition of Algorithm~\ref{alg:MLOP_Alg}.
Observe that in this case the definitions of $ T_A(\cdot)$ and $T_P(\cdot)$ imply that:
\begin{align}
   T_A(Z)=T_P(\Alg(P)\cup Z). \label{UB-eq6}
\end{align}
Applying the induction hypothesis to $P$ yields that: $ T_P(\Alg(P)\cup Z)=((\Alg(P)\cup Z)\cap P)\cup h(P)$.
Hence,
\begin{align}
    T_P(\Alg(P)\cup Z) & = ((\Alg(P)\cap P)\cup (Z\cap P))\cup h(P) \label{UB-eq7}\\
    & = ((\Alg(P)\cap P)\cup (Z\cap A))\cup h(P) \label{UB-eq8} \\
    & = (Z\cap A)\cup h(A), \label{UB-eq9}
\end{align}
where:
$(1)$ equality~\eqref{UB-eq7} follows from standard boolean arithmetics, i.e., $ (\Alg(P)\cup Z)\cap P = (\Alg(P)\cap P)\cup (Z\cap P)$;
$(2)$ equality~\eqref{UB-eq8} follows since $ A=P\setminus (\Alg(P)\cap P)$, implying that $ (\Alg(P)\cap P)\cup (Z\cap P) = (\Alg(P)\cap P)\cup (Z\cap A)$; and
$(3)$ equality~\eqref{UB-eq9} follows since $ h(A)=h(P)\cup (\Alg(P)\cap P)$, as the left child of $P$, namely $ \Alg(P)\cap P$, precedes the right child of $P$, namely $A$, in the final ordering.
Combining~\eqref{UB-eq6} and~\eqref{UB-eq9} concludes the proof of the inductive step in the case where $A$ is a right child.
\end{proof}

Equipped with Lemma~\ref{clm:TA_general_identity}, Lemma~\ref{lem:alg_val_upper_bound} can now be proved.
\begin{proof}[Proof of Lemma~\ref{lem:alg_val_upper_bound}]
Fix \(A\in \mathcal{A}\) and consider its term \(f\!\left( \left(\Alg(A) \cap A\right) \cup \hS(A) \right)\) in~\eqref{objective-decomposition}.
Applying Lemma~\ref{clm:TA_general_identity} with $Z = \Alg(A)$ yields:
\begin{align}
f\!\left( \left(\Alg(A) \cap A\right) \cup \hS(A) \right) 
= f\bigl(T_A(\Alg(A))\bigr) 
\le \sum_{i=0}^{d(A)} f\bigl(\Alg(A^{(i)})\bigr). \label{UB-ineq5}
\end{align}
The inequality in~\eqref{UB-ineq5} follows from submodularity and non-negativity of $f$ since both imply that $ f(S\cup T)\leq f(S)+f(T)$ and $ f(S\cap T)\leq f(S)+f(T)$ for every $ S,T\subseteq N$.
Specifically, the recursive definition of $T_A(\Alg(A)) $ implies that $T_A(\Alg(A)) $ is obtained via successive union and intersection operations with $ \Alg(A^{(0)}), \Alg(A^{(1)}),\ldots,\Alg(A^{(d(A))})$, where each $ \Alg(A^{(i)})$ appears exactly once in this sequence of operations.

Therefore, plugging~\eqref{UB-ineq5} into~\eqref{objective-decomposition} provides:
\begin{align}
\val(\pi)
\le
\sum_{A\in \mathcal{A}} \sum_{i=0}^{d(A)} f\bigl(\Alg(A^{(i)})\bigr).\label{UB-ineq6}
\end{align}
Observe that in the double sum~\eqref{UB-ineq6}, each term \(f(\Alg(A))\) is counted once for every node
\(B\in\mathcal{A}\) such that \(A\) is an ancestor of \(B\) in the decomposition tree,
i.e., for every node contained in the subtree rooted at \(A\).
Since Algorithm \ref{alg:MLOP_Alg} partitions subsets until reaching singletons, the subtree rooted at \(A\)
contains exactly \(|A|\) leaves.
As the decomposition tree is a full binary tree and \(\mathcal{A}\) contains only non-leaf nodes,
this subtree therefore contains exactly \(\lvert A\rvert-1\) internal nodes.
Hence, each term \(f(\Alg(A))\) appears exactly \(\lvert A\rvert-1\) times in the sum.
Thus,
\[
\val(\pi)
\le
\sum_{A\in\mathcal{A}} (\lvert A\rvert -1)\cdot f\bigl(\Alg(A)\bigr).
\]
The proof is concluded by observing that $ \lvert A\rvert-1\leq \lvert A\rvert$.
\end{proof}

\paragraph*{Proving Lemma~\ref{lem:structural-decomposition}.}
The proof of the lower bound Lemma~\ref{lem:structural-decomposition} provides boils down to the observation that one can choose sufficiently many distinct solutions to the relevant \SBClong instances that correspond to distinct prefixes of an optimal ordering for \MLOP.
\begin{proof}[Proof of Lemma~\ref{lem:structural-decomposition}]
Let $\pi^\star$ be an optimal ordering of $N$, i.e., $ \val(\pi^{\star})=\OPT$.
Recall that:
\begin{align}
\val(\pi^{\star})
\;=\;
\OPT
\;=\;
\sum_{t=1}^{n-1} f\!\left(N_{\pi^\star,t}\right), \label{UB-eq3}
\end{align}
where $N_{\pi^\star,t}$ denotes the prefix of $\pi^{\star}$ consisting of its first
$t$ elements: $N_{\pi^\star,t}=\{ u\in N\colon \pi^{\star}(u)\leq t\} $.

Fix $X_i$ for some $i \in [\ell]$, and for every $j = 1,\dots,\lvert X_i\rvert$, let $C_{i,j} \subseteq N$ denote the prefix of $\pi^\star$
that contains all elements appearing in $\pi^\star$ no later than the $j$\textsuperscript{th}
element of $X_i$, where the elements of $X_i$ are ordered according to their
appearance order in $\pi^\star$.
Thus, if the $j$\textsuperscript{th} element of $X_i$ with respect to the order $\pi^{\star}$ is $u$, and $ \pi^{\star}(u)=k$, then $ C_{i,j}=N_{\pi^{\star},k}$.

Let $\mathcal{C}_i \;=\; \{ C_{i,1}, C_{i,2}, \dots, C_{i,\lvert X_i\rvert} \} $.
Note that each subset $C_{i,j}$ is a distinct prefix of $\pi^\star$ and thus
corresponds to a distinct cut in the right-hand side of~\eqref{UB-eq3}.
Hence,
\begin{align} \sum_{j=1}^{\lvert X_i\rvert} f(C_{i,j}) \geq \sum_{j=\left\lceil \lvert X_i\rvert/4 \right\rceil}^{\left\lfloor 3\lvert X_i\rvert/4 \right\rfloor} f(C_{i,j}) \geq \sum_{j=\left\lceil \lvert X_i\rvert/4 \right\rceil}^{\left\lfloor 3\lvert X_i\rvert/4 \right\rfloor} f(S_i^\star) \geq \frac{\lvert X_i\rvert}{6}\cdot f(S_i^\star).\label{UB-ineq3}
\end{align}
In the above:
$(1)$ the first inequality holds since $f$ is nonnegative;
$(2)$ the second inequality holds since, for every $j \in \{\lceil \lvert X_i\rvert/4 \rceil,\dots,\lfloor 3\lvert X_i\rvert/4 \rfloor\}$, the subset $C_{i,j}$ contains exactly $j$ elements of $X_i$ by construction, and therefore is a feasible solution to the \SBClong instance defined by $ (N,f,w)$ with $w=\mathbf{1}_{X_i}$ and $b=\nicefrac{1}{4}$ (implying that $C_{i,j}$ is a feasible solution for the same instance for which $S_i^{\star}$ is an optimal solution and thus $f(S_i^\star) \le f(C_{i,j})$); and
$(3)$ the third and last inequality holds since $ \left\lfloor 3\lvert X_i\rvert/4 \right\rfloor
-
\left\lceil \lvert X_i\rvert/4 \right\rceil
+ 1\geq \lvert X_i\rvert/6$ whenever $\lvert X_i\rvert\geq 2$.

Since the subsets $X_1,\dots,X_\ell$ are pairwise disjoint, the collections
$\mathcal{C}_1,\dots,\mathcal{C}_\ell$ consist of distinct prefixes of $\pi^\star$.
Therefore,
\begin{align}
\OPT
\; = \;
\sum_{t=1}^{n-1} f(N_{\pi^\star,t})
\; \geq \;
\sum_{i=1}^{\ell}\sum_{j=1}^{\lvert X_i\rvert} f(C_{i,j}).
\label{UB-ineq4}
\end{align}
Combining~\eqref{UB-ineq4} with~\eqref{UB-ineq3} concludes the proof with $c=6$.
\end{proof}


\begin{remark*}
It is worth noting that if particular classes of submodular functions admit an approximation algorithm for the corresponding balanced cut problem with a better guarantee than that of Theorem~\ref{thm:balanced_cut_algorithm}, then our framework of recursive balanced cutting using global cuts, as appearing in Algorithm~\ref{alg:MLOP_Alg}, yields correspondingly improved guarantees.

Moreover, the aforementioned framework using global cuts continues to apply even in settings where no natural local restriction of the objective function to recursive subinstances is available.
\end{remark*}

\section{Tight Information-Theoretic Hardness}\label{sec:hardness}
Given an ordering $\pi$ of $N$, denote by $R_{\pi,i} \triangleq N \setminus N_{\pi,i}= \{ u\in N\colon \pi(u)>i\} $ the suffix of length $n-i$ of $\pi$.
Let $\delta > 0$ be a parameter to be chosen later and let $\pi$ be a uniformly random ordering of $N$.
Define the following two set functions
$f_1, f_2 : 2^N \to \mathbb{R}$, where $f_1$ is deterministic and $f_2$ is random, since it depends on $\pi$:
\begin{align}
f_1(S)
&\triangleq \min\left\{
\lvert S\rvert/2,
|\bar{S}|/2
\right\}, \label{f1}\\
f_2(S)
&\triangleq \min\left\{
f_1(S),
\min \left\{\bigl(i/2 + \delta\bigr) +\lvert S \cap R_{\pi,i}\rvert - \lvert S\rvert/2
\;\colon\;
i = 1,\dots,n-1
\right\}\right\}.\label{f2}
\end{align}
Intuitively, hiding a random ordering $\pi$ through its suffixes allows $f_2$ to include, for each prefix $N_{\pi,i}$ appearing in the \MLOP objective, a term with a significantly low value. 

\begin{remark*}
The construction above is inspired by the indistinguishability constructions of Svitkina and Fleischer~\cite{svitkina2011submodular} for \SBClong and the subsequent hardness construction of~\cite{iwata2012approximating}, which yields a hardness of $2$ for \MLOP.
In particular, the deterministic function $f_1$ is identical to the one used in these earlier constructions, while the randomized function $f_2$ extends the previous approach from hiding only a constant number of random sets to hiding an entire random ordering through all of its suffixes.
\end{remark*}


First, we establish a general lemma that provides a sufficient condition for submodularity that applies to a broad class of set functions, including those used in our hardness construction.
\begin{lemma}\label{lem:g_submodular_general_form}
Let $U_m \subseteq U_{m-1} \subseteq \cdots \subseteq U_0 \subseteq N$ be a nested family of subsets of $N$,
and let $\alpha_0,\alpha_1,\ldots,\alpha_m \in \mathbb{R}$.
Define a set function $g\colon2^N \to \mathbb{R}$ as follows:
\[
g(S) \triangleq \min_{0 \le i \le m}\left\{ \alpha_i + \lvert S \cap U_i\rvert \right\}.
\]
Then, $g$ is submodular.
\end{lemma}
\begin{proof}
For every $i\in \{ 0,\ldots,m\}$, define $ g_i\colon 2^N\rightarrow \mathbb{R}$ as follows: $g_i(S)\triangleq \alpha_i + \lvert S\cap U_i\rvert$. Consider the following two observations:
\begin{enumerate}[(i)]
    \item For every $ i\in \{ 0,\ldots,m\}$ the function $ g_i$ is submodular: this is true since $g_i$ is modular and thus it is also submodular. \label{observation1}
    \item For every $ i<j$ the function $ (g_i-g_j)(S)\triangleq g_i(S)-g_j(S)$ is monotone non-decreasing: this is true since $ (g_i-g_j)(S)=(\alpha_i-\alpha _j) + \lvert S\cap (U_i\setminus U_j)\rvert$ as $ U_j\subseteq U_i$. \label{observation2}
\end{enumerate}
We prove that $g(S)$ is submodular by induction on $m$.
The base case is when $m=0$, implying that: $g=g_0$.
Since observation~\ref{observation1} above implies that $g_0$ is submodular, the base case holds.

Focusing on the inductive case of $m>0$, consider the following two functions: $ \min _{0\leq i\leq m-1}\{ g_i\}$ and $g_m$.
The induction hypothesis implies that $ \min _{0\leq i\leq m-1}\{ g_i\}$ is submodular, and observation~\ref{observation1} implies that $g_m$ is submodular.
Moreover, one can note that $ \min _{0\leq i\leq m-1}\{ g_i\} - g_m$ is monotone non-decreasing.
The reason for the latter is that:
\begin{align}
    \min_{0 \leq i\leq m-1}\{ g_i\}-g_m =  \min_{0\leq i\leq m-1}\{ g_i -g_m\}, \label{inductive-step}
\end{align}
and the right-hand side of~\eqref{inductive-step} is a point-wise minimum of monotone non-decreasing functions (recalling that observation~\ref{observation2} implies that $ g_i-g_m$ is monotone non-decreasing for every $ i<m$).
Therefore, the right-hand side of~\eqref{inductive-step} is a monotone non-decreasing function on its own since it is the point-wise minimum of monotone non-decreasing functions.

All the conditions of Lemma~\ref{lem:lovasz} are satisfied for the two functions $ \min _{0\leq i\leq m-1}\{g_i \}$ and $g_m$, implying that $ \min \{ \min _{0\leq i\leq m-1}\{ g_i\},g_m\}=\min _{0\leq i\leq m}\{ g_i\}$ is a submodular function.
\end{proof}

The following lemma proves that $f_1$ and $f_2$, regardless of the outcome of the random choice $\pi$, are both non-negative and submodular (as required).

\begin{lemma}\label{lem:functions_are_submodular_nonnegative}
For every $\delta>0$ and every ordering $\pi$ of $N$, $f_1$ and $f_2$ are non-negative and submodular.
\end{lemma}

\begin{proof}[Proof of Lemma~\ref{lem:functions_are_submodular_nonnegative}]
We start with $f_1$. By definition~\eqref{f1}, $f_1(S)$ is non-negative. For submodularity, note that $f_1$~\eqref{f1} can be equivalently written as:
\begin{align}
    f_1(S)
=\min\{|S|,n/2\}-|S|/2.\label{submodular-f1-1}
\end{align}
Both terms in the minimum of \eqref{submodular-f1-1} are of the form $\alpha_i + \lvert S \cap R_{\pi,i}\rvert$:
$(1)$ the term $|S|$ corresponds to $i=0$ with $R_{\pi,0}=N$ and $\alpha_0=0$;
$(2)$ the term $n/2$ corresponds to $i=n$ with $R_{\pi,n}=\varnothing$ and $\alpha_n=n/2$.
Since $\varnothing=R_{\pi,n}\subseteq R_{\pi,0} = N$ forms a nested chain, all conditions of Lemma~\ref{lem:g_submodular_general_form} are satisfied.
This yields that the minimum in~\eqref{submodular-f1-1} is a submodular function.
Since $|S|/2$ is modular, subtracting it preserves submodularity, and hence $f_1$ is submodular.

Let us now focus on $f_2$.
Fix any $\delta>0$ and any ordering $\pi$ of $N$.
First, we prove that $f_2$ is non-negative.
Since $f_1$ is non-negative, it suffices to show that for every $i\in \{ 1,\ldots,n-1\}$:
\begin{align}
\left(i/2 + \delta\right) +\lvert S \cap R_{\pi,i}\rvert - |S|/2 \ge 0. \label{submodular-f2-1}
\end{align}
Fixing $i$ and substituting $|S|$ with $\lvert S\cap R_{\pi,i}\rvert + \lvert S\cap N_{\pi,i}\rvert$ in~\eqref{submodular-f2-1} yields:
\begin{align}
\left(i/2 + \delta\right) +\lvert S \cap R_{\pi,i}\rvert - \lvert S\rvert/2 = 
\left(i/2 + \delta\right) +\lvert S \cap R_{\pi,i}\rvert/2
- \lvert S\cap N_{\pi,i}\rvert/2 \geq i/2 - \lvert S\cap N_{\pi,i}\rvert/2.\label{submodular-f2-2}
\end{align}
The right-hand side of~\eqref{submodular-f2-2} is non-negative since $\lvert S\cap N_{\pi,i}\rvert \le \lvert N_{\pi,i}\rvert = i$.

Second, we prove that $f_2$ is submodular.
To this end, note that $f_2$~\eqref{f2} can be equivalently written as:
\begin{align}
f_2(S) &= \min\left\{
|S|,n/2,
\min \left\{\left(i/2 + \delta\right) +|S \cap R_{\pi,i}|\;\colon \;
i = 1,\dots,n-1
\right\}\right\} - |S|/2. \label{submodular-f2-3}
\end{align}
Each of the $n+1$ terms in the minimum of~\eqref{submodular-f2-3} is of the form $ \alpha _i+|S\cap R_{\pi,i}|$:
$(1)$ the term $|S|$ corresponds to $i=0$ with $ R_{\pi,0}=N$ and $\alpha_0=0$;
$(2)$ the term $ n/2$ corresponds to $i=n$ with $ R_{\pi,n}=\varnothing$ and $ \alpha _n=n/2$; and
$(3)$ the remainder of the terms correspond to $i\in \{ 1,\ldots,n-1\}$ with $ R_{\pi,i}$ and $ \alpha_i=i/2+\delta$.
Noting that $ \varnothing=R_{\pi,n}\subseteq R_{\pi,n-1}\subseteq \ldots \subseteq R_{\pi,1}\subseteq R_{\pi,0}=N$ forms a nested chain, all conditions of Lemma~\ref{lem:g_submodular_general_form} are satisfied.
This yields that the minimum in~\eqref{submodular-f2-3} is a submodular function.
Since $|S|/2$ is modular, subtracting it preserves submodularity, and hence $f_2$ is submodular. 
\end{proof}

The following lemma shows that for a suitable choice of $\delta$, $f_1$ and $f_2$ satisfy the probabilistic requirement of the indistinguishability lemma, i.e., Lemma~\ref{lem:general_Indistinguishability_def}.
In what follows, recall that $f_1$~\eqref{f1} is deterministic whereas $ f_2$~\eqref{f2} is random and its randomness is determined by a uniform random ordering $\pi$ of $N$.

\begin{lemma}\label{lem:functions_satisfy_cond}
Assume that $\delta = \omega(\sqrt{n \ln n})$.
Then, for every $ S\subseteq N$: $ \Pr_{\pi}[f_1(S)\neq f_2(S)]\leq n^{-\omega(1)}$.
\end{lemma}
\begin{proof}
Fix $S\subseteq N$. By the definition of $f_2$~\eqref{f2}:
\begin{align}
\Pr\nolimits_\pi[f_1(S)\neq f_2(S)]
=
\Pr\nolimits_\pi\left[\exists i\in\{ 1,\ldots,n-1\}: (i/2+ \delta) + \lvert S \cap R_{\pi,i}\rvert - \lvert S\rvert/2 < f_1(S)\right]. \nonumber 
\end{align}
For each $i\in\{ 1,\ldots,n-1\}$, let
$p_i(S) \triangleq \Pr_\pi[\bigl(i/2+ \delta\bigr) + \lvert S \cap R_{\pi,i}\rvert - \lvert S\rvert/2 < f_1(S)]$.
We note that proving, for all $i$, that $p_i(S) \le n^{-\omega(1)}$, suffices to conclude the proof.
The reason is that the union bound over all possible $i$ gives:
\(\Pr[f_1(S)\neq f_2(S)] \le \sum _{i=1}^{n-1}p_i(S)\leq (n-1)\cdot n^{-\omega(1)}\leq   n^{-\omega(1)}\), as required.

We start by showing that $p_i(S)$ is maximized when $\lvert S\rvert = n/2$.
Suppose that $\lvert S\rvert\ge n/2$.
In this case $f_1(S) = |\bar{S}|/2$, and hence:
\begin{align}
p_i(S) = \Pr\nolimits_\pi\left[ (i/2+ \delta) + \lvert S \cap R_{\pi,i}\rvert - \lvert S\rvert/2 < \lvert\bar{S}\rvert/2\right] = \Pr\nolimits_\pi\left[\lvert S \cap R_{\pi,i}\rvert < n/2 - (i/2+ \delta) \right]. \label{indis-prob2}
\end{align}
Thus, removing an element from $S$ can only increase the probability on the right-hand side of~\eqref{indis-prob2}.
Suppose that $\lvert S\rvert \leq n/2$. In this case, $f_1(S) = |S|/2$, and hence:
\begin{align}
p_i(S) = \Pr\nolimits_\pi\left[ (i/2+ \delta) + \lvert S \cap R_{\pi,i}\rvert - \lvert S\rvert/2 < \lvert S\rvert/2\right] = \Pr\nolimits_\pi\left[ i/2+ \delta <  \lvert S \cap N_{\pi,i}\rvert\right], \label{indis-prob3}
\end{align}
where the last equality follows since $\vert S\cap R_{\pi,i}\rvert + \lvert S\cap N_{\pi,i}\rvert = \lvert S\rvert$.
Thus, adding an element to $S$ can only increase the probability on the right-hand side of~\eqref{indis-prob3}.
Therefore, it suffices to consider sets $S$ of size $n/2$.
Assume henceforth that $\lvert S\rvert = n/2$.

Focusing on $p_i(S)$, note that $R_{\pi,i}$ is a uniformly random subset of $N$ of size $n-i$.
Hence, the quantity $|S\cap R_{\pi,i}|$ counts how many of these $n-i$ sampled
elements belong to $S$. As $\lvert S\rvert=n/2$, this corresponds to sampling without
replacement from a population consisting of $n/2$ ones (corresponding to
elements of $S$) and $n/2$ zeros (corresponding to elements of $\bar{S}$).
By Hoeffding's inequality for sampling
without replacement (see, e.g.,~\cite{hoeffding1963probability,bardenet2013concentration}),
\begin{align*}
    p_i(S) &= \Pr\nolimits_{\pi}\left[\lvert S\cap R_{\pi,i}\rvert - (n-i)/2 < -\delta\right]
    \leq \exp\left(-2\delta^2/(n-i)\right)
    \leq \exp\left(-2\delta^2/n\right)
    \leq n^{-\omega(1)},
\end{align*}
where the last inequality follows since $\delta^2/n = \omega(\ln n)$ (recall that $ \delta=\omega(\sqrt{n\ln{n}})$). This establishes the bound on $p_i(S)$ assumed above, and completes the proof.
\end{proof}

We are now ready to prove Theorem~\ref{thm:lower_bound}. 

\begin{proof}[Proof of Theorem~\ref{thm:lower_bound}]
Consider \MLOP with the two functions $f_1$ and $f_2$ defined earlier. Throughout this proof, by a slight abuse of notation, we use $f_2$ to denote both
the random function and an arbitrary fixed realization of it.

First, focus on $f_1$. Let $\OPT_1$ denote the optimal value of \MLOP with the function $f_1$.
For any ordering $\pi$ and any $i \in \{ 1,\ldots,n-1\}$:
$f_1(N_{\pi,i}) = \min\{ i/2, (n-i)/2 \}$,
which depends only on $i$.
Hence, the \MLOP objective with $f_1$ does not depend on the ordering $\pi$, yielding:
\begin{align}
\OPT_1
= \sum_{i=1}^{n-1} \min\left\{ i/2, (n-i)/2 \right\}
\ge n^2/16. \nonumber
\end{align}

Second, consider $f_2$. Fix an arbitrary realization of $f_2$, and let $\OPT_2$ denote
the optimal value of \MLOP with $f_2$. The ordering $\pi$ that defines $f_2$
therefore satisfies $\OPT_2\leq \sum _{i=1}^{n-1}f_2(N_{\pi,i})$.
Note that for every $ i\in \{ 1,\ldots,n-1\}$: $ f_2(N_{\pi,i})\leq \delta$.
The reason for the latter is that $N_{\pi,i}\cap R_{\pi,i}=\varnothing$ and $\lvert N_{\pi,i}\rvert = i$, and thus:
\begin{align}
f_2(N_{\pi,i})\leq \left(i/2+\delta\right)
+ \lvert N_{\pi,i}\cap R_{\pi,i}\rvert
- \lvert N_{\pi,i}\rvert/2
= \delta.\nonumber
\end{align}
Therefore, $ \OPT_2\leq \sum _{i=1}^{n-1}f_2(N_{\pi,i})\leq n\delta$.

We conclude that for all possible realizations of $f_2$:
\begin{align}
   \OPT_1 \geq n^2/16 = \gamma \cdot n\delta \geq \gamma \cdot \OPT_2, \label{optima}
\end{align}
where the equality follows from setting $ \gamma=n/(16\delta)$.

Choosing $\delta =\omega(\sqrt{n\ln{n}})$ implies that:
$(1)$ the indistinguishability requirement of Lemma~\ref{lem:general_Indistinguishability_def} is satisfied by Lemma~\ref{lem:functions_satisfy_cond} and our choice of $\delta$; and
$(2)$ the optima ratio requirement of Lemma~\ref{lem:general_Indistinguishability_def} is satisfied with  $\gamma=n/(16\delta)$~\eqref{optima} .
Therefore, the two requirements of Lemma~\ref{lem:general_Indistinguishability_def} are satisfied.
Thus, the above and Lemma~\ref{lem:functions_are_submodular_nonnegative}, which establishes that $f_1$ and $f_2$ are always non-negative and submodular, conclude the proof of the theorem since $ \gamma = o(\sqrt{n/\ln{n}})$.
\end{proof}

\section{Funding}
Evan Abboud and Roy Schwartz received funding from the European Union’s Horizon 2020 research and innovation program under
grant agreement no. 852870-ERC-SUBMODULAR and ISF grant 3264/25.

\bibliography{sources}

\appendix
\section{Missing Proofs: Algorithm~\ref{alg:MLOP_Alg}}
\label{app:results}

\begin{lemma}\label{lem:final_sum_upper_bound}
Let $\beta\in(0,1)$ be a fixed constant and let
$D \le \sfrac{\ln(n/2)}{\ln(1/\beta)}$. Then,
\[
\sum_{i=0}^{D} \sqrt{\frac{\beta^i\, n}{\ln(\beta^i\, n)}}
= O\!\left(\sqrt{\frac{n}{\ln n}}\right).
\]
\end{lemma}

\begin{proof}
From the assumption on $D$, we have $\beta^D n \ge 2$, and since $0<\beta<1$, this implies
$\beta^i n \ge 2$ for all $0\le i\le D$, hence $\ln(\beta^i n)\ge \ln 2$.

Let
\[
I := \left\lfloor \frac{\ln n}{2\ln(1/\beta)} \right\rfloor. 
\]
Then $\beta^I n \ge \sqrt{n}$ and $\beta^{I+1} n \le \sqrt{n}$.

For $0\le i\le \min\{I,D\}$, we have $\ln(\beta^i n)\ge\ln(\beta^I n)\ge \tfrac12 \ln n$, and therefore
\[
\sqrt{\frac{\beta^i n}{\ln(\beta^i n)}}
\le
\sqrt{\frac{2n}{\ln n}}\,\bigl(\sqrt{\beta}\bigr)^i.
\]
Summing and using that $\beta$ is constant,
\[
\sum_{i=0}^{\min\{I,D\}} \sqrt{\frac{\beta^i n}{\ln(\beta^i n)}}
\le
\sqrt{\frac{2n}{\ln n}} \sum_{i=0}^{\infty} (\sqrt{\beta})^i
= O\!\left(\sqrt{\frac{n}{\ln n}}\right).
\]

If $D>I$, then for $i\ge I+1$ we have $\beta^i n \le \beta^{I+1} n\leq \sqrt{n}$ and $\ln(\beta^i n)\ge \ln 2$, hence
\[
\sqrt{\frac{\beta^i n}{\ln(\beta^i n)}}
\le
\sqrt{\frac{n}{\ln 2}}\,(\sqrt{\beta})^i.
\]
Thus,
\[
\sum_{i=I+1}^{D} \sqrt{\frac{\beta^i n}{\ln(\beta^i n)}}
\le
\sqrt{\frac{n}{\ln 2}} \sum_{i=I+1}^{\infty} (\sqrt{\beta})^i
= O\!\left(n^{1/4}\right)
= o\!\left(\sqrt{\frac{n}{\ln n}}\right).
\]
Combining the two parts yields the claim.
\end{proof}

\section[Submodular b-Balanced Cut: Improving Svitkina and Fleischer]%
{Submodular $b$-Balanced Cut: Improving Svitkina and Fleischer~\cite{svitkina2011submodular}}
\label{app:ssc_probability}
This appendix provides the technical details required for proving
Theorem~\ref{thm:balanced_cut_algorithm}, which provides an approximation algorithm
for the \SBClong\ (\SBC) problem.
We adapt the analysis of
Svitkina and Fleischer~\cite{svitkina2011submodular} so that the approximation
guarantee depends on the number of elements of positive weight, rather than the
size of the ground set.
Moreover, we note that essentially the same algorithm and analysis of~\cite{svitkina2011submodular}, albeit with a few minor changes, apply to obtain the improved result.

At the core of the algorithm is a \WSSClong (\WSSC) subroutine used by
the \SBC algorithm.
We show that, in the weighted setting, this subroutine admits an
$O(\sqrt{k/\ln k})$ guarantee, where
$k = |\{v\in V : w(v)>0\}|$.
\paragraph{Organization.}
We first present an algorithm for \WSSC that takes as input an upper bound $B$ on the
optimal value and a success probability $p$, and returns with a probability of at least $p$ a cut whose objective value is at most $O(\sqrt{k/\ln k})\cdot B$.
We then show how this guarantee yields an approximation algorithm for \WSSC.
Next, we present the \SBC\ algorithm and its analysis, completing the proof of
Theorem~\ref{thm:balanced_cut_algorithm}.
We conclude by discussing probability amplification for \SBC\ and explaining how
these randomized guarantees propagate to the final analysis of our \MLOP
algorithm, i.e., Algorithm~\ref{alg:MLOP_Alg}.

\paragraph{\WSSClong.}
In the \WSSC problem we are given a ground set $V$ of size $n$, a nonnegative submodular function $f:2^V\to\mathbb{R}_{\ge 0}$, and a nonnegative weight function $w:V\to\mathbb{R}_{\ge 0}$.
Define a demand between every unordered pair $\{u,v\}\subseteq V$ as follows:
$d_{u,v}\triangleq w(u)\cdot w(v)$.
The goal in \WSSC is to find a set $S\subseteq V$ minimizing the ratio between the cut
value and the total separated demand:
$f(S)/(w(S)\cdot w(\bar S))$.

Let $K \triangleq \{v \in V : w(v) > 0\}$ denote the support of the weight function $w$, and let $k \triangleq |K|$.
A description of the algorithm appears in Algorithm~\ref{alg:WSSC}.
There are two things to note.
First, the constant $c>0$ in the number of iterations is an absolute constant to be determined later by the analysis.
Second, Algorithm~\ref{alg:WSSC} requires a black-box algorithm for unconstrained submodular minimization, a problem that is known to be polynomial time solvable.
There are various algorithms achieving this, starting from the classic result of Gr\"{o}tschel, Lov\'{a}sz and Schrijver~\cite{GLS81}, and subsequent faster algorithms, e.g.,~\cite{IFF01,SCHRIJVER2000346}.

\begin{algorithm}[th]
\caption{\WSSClong}
\label{alg:WSSC}
\begin{algorithmic}[1]
\Require $V,f,w$, and parameters $B,p$
\Ensure $T\subseteq V$ (or \textsf{fail})

\State $\alpha \gets 4\sqrt{\frac{k}{\ln k}}\cdot B$

\For{$\frac{8k^3}{c}\ln\!\big(\frac{1}{1-p}\big)$ iterations}

    \State initialize $\tilde{w}(v)\gets 0$ for all $v\in V$
    \State sample $S\subseteq V$ by including each $v\in V$ independently with a probability of $\nicefrac{1}{2}$
    \State for all $u\in S$: $\tilde{w}(u) \gets \sum_{v\notin S}d_{u,v}$
    \State for all $u\notin S$: $\tilde{w}(u) \gets -\sum_{v\in S}d_{u,v}$
    \State compute a set $T\subseteq V$ minimizing the submodular function
$f(T)-\alpha\sum_{v\in T}\tilde w(v)$
    \If{$f(T)-\alpha\sum_{v\in T}\tilde{w}(v) < 0$}
        \State \Return $T$
    \EndIf
\EndFor
\State \Return \textsf{fail}
\end{algorithmic}
\end{algorithm}

\begin{lemma}\label{lem:wssc_val_bound}
Fix a set $S \subseteq V$ sampled in one iteration of
Algorithm~\ref{alg:WSSC}. If the algorithm finds a set
$T \subseteq V$ such that
$f(T) - \alpha \sum_{v \in T} \tilde{w}(v) < 0$,
then:
\[
\frac{f(T)}{w(T)\cdot w(\bar T)} < \alpha .
\]
\end{lemma}
\begin{proof}
Let $a_1 = w(T\cap S)$, $a_2 = w(T\setminus S)$, $b_1 = w(\bar T\cap S)$, and
$b_2 = w(\bar T\setminus S)$.
Then $w(K\setminus S)=a_2+b_2$ and $w(K\cap S)=a_1+b_1$.
By definition of Algorithm~\ref{alg:WSSC}:
\begin{align*}
\sum_{v\in T}\tilde{w}(v)
&=
\sum_{v\in T\cap S} w(v)\cdot w(K\setminus S)
\;-\;
\sum_{v\in T\setminus S} w(v)\cdot w(K\cap S) \nonumber\\
&= a_1(a_2+b_2) - a_2(a_1+b_1) \nonumber\\
&\le (a_1+a_2)(b_1+b_2) \nonumber\\
&= w(T)\cdot w(\bar T). 
\end{align*}
Using this inequality and the assumption of the lemma, we obtain:
\[
f(T) - \alpha\, w(T)\,w(\bar T)
\le
f(T) - \alpha\sum_{v\in T}\tilde{w}(v)
< 0.
\]
Since $f$ is non-negative, it follows that $w(T)\,w(\bar T)>0$.
Rearranging yields $\frac{f(T)}{w(T)\,w(\bar T)}<\alpha$.
\end{proof}

We now state a concentration bound that will be used in the analysis.

\begin{theorem}[Svitkina and Fleischer~\cite{svitkina2011submodular}]
\label{thm:svitkina_concentration_bound}
Suppose that $m$ elements are selected independently, with probability $0<q<1$ each.
Then for $0 \le \varepsilon < \frac{1-q}{q}$, the probability that exactly
$\left\lceil qm(1+\varepsilon)\right\rceil$ elements are selected is at least
$cq \cdot m^{-3/2}\cdot \exp\!\left(-\,\varepsilon^2\cdot \frac{qm}{1-q}\right)$, for some constant $c>0$.

\end{theorem}

Assume the instance is feasible, and let $U^\star\subseteq V$ be a solution of value strictly less than $B$,
that is, $f(U^\star)/D^\star < B$, where $D^\star = w(U^\star)\cdot w(V \setminus U^\star)$.
Let $m = |U^\star\cap K|$.

\begin{lemma}
\label{lem:wssc_iteration_success}
For each iteration of Algorithm~\ref{alg:WSSC},
\[
\Pr\nolimits_S\!\left[
\sum_{v\in U^\star}\tilde{w}(v)
\ \ge\
\frac{1}{4}\sqrt{\frac{\ln k}{k}}\cdot D^\star
\right]
\ \ge\
\frac{c}{8k^3}.
\]
\end{lemma}

\begin{proof}
Let $\varepsilon \triangleq \sqrt{\ln k / k}$.
Define $E$ as the event that
$\lvert U^\star \cap K \cap S\rvert \ge \tfrac{m}{2}(1+\varepsilon)$,
where $S$ is the random set sampled in this iteration.
We bound the desired probability by conditioning on $E$:
\begin{align}
\Pr\nolimits_S\!\left[
\sum_{v\in U^\star}\tilde{w}(v)
\ge
\frac{\varepsilon}{4}\, D^\star
\right]
\ge
\Pr\nolimits_S\!\left[
\sum_{v\in U^\star}\tilde{w}(v)
\ge
\frac{\varepsilon}{4}\, D^\star
\ \Bigm|\ E
\right]
\cdot
\Pr\nolimits_S[E].
\label{eq:wssc_conditioning}
\end{align}
By Theorem~\ref{thm:svitkina_concentration_bound}, we have
$\Pr_S[E] \ge \frac{c}{2}\cdot k^{-5/2}$, using $m \le k$.

All probabilities and expectations in the remainder of the proof
are conditioned on the event $E$.
We now analyze the expected value of $\sum_{v\in U^\star}\tilde{w}(v)$.
By construction, $\sum_{v\in U^\star}\tilde{w}(v)$ can be written as a sum of signed demands $d_{u,v}=w(u)w(v)$ over all pairs $\{u,v\}\subseteq K$ that are separated by the cut $(S,\bar S)$ and have exactly one endpoint in $U^\star$, i.e., $ \sum_{v\in U^\star}\tilde{w}(v)=\sum _{v\in U^{\star},u\notin U^{\star}}(1_{\{ v\in S,u\notin S\}}d_{u,v}-1_{\{ v\notin S, u\in S\}}d_{u,v})$.
Fix a pair $\{u,v\}\subseteq K$ with exactly one endpoint in $U^\star$, and assume without loss of generality that $u\in U^\star$ and $v\in V\setminus U^\star$. Let $p_u \triangleq \Pr[u\in S ]$ and $p_v \triangleq \Pr[v\in S ]$. Since the event $E$ depends only on the random choices of vertices in $U^\star\cap K$, and $v\notin U^\star$, the event of $v$ being in $S$ is independent of
$E$, and therefore $p_v=1/2$. By symmetry among elements of $U^\star \cap K$, conditioning on the event $E$ implies that each element of $U^\star\cap K$ is included in $S$ with the same probability, equal to $p_u$. Therefore, $\mathbb{E}\!\left[\,|U^\star\cap K\cap S| \right]=m\cdot p_u$, and it follows from the event $E$ that $p_u \geq (1+\varepsilon)/2$.
We have:
\[
\Pr[u\in S \wedge v\notin S]
=
p_u(1-p_v)
\ge
\frac{1+\varepsilon}{4},
\]
\[
\Pr[u\notin S \wedge v\in S]
=
(1-p_u)p_v
\le
\frac{1-\varepsilon}{4}.
\]
Then, the expected contribution of this demand pair to
$\sum_{v\in U^\star}\tilde{w}(v)$ is:
\[
\Pr[u\in S \wedge v\notin S]\cdot d_{u,v}
\;+\;
\Pr[u\notin S \wedge v\in S]\cdot (-d_{u,v})
\;\ge\;
\frac{\varepsilon}{2}\cdot d_{u,v}.
\]
By linearity of expectation,
\[
\mathbb{E}\!\left[\sum_{v\in U^\star}\tilde{w}(v)\right]
\ge
\frac{\varepsilon}{2}\sum_{u\in U^\star\cap K}\sum_{v\in (V\setminus U^\star)\cap K} d_{u,v}
=
\frac{\varepsilon}{2}\, w(U^\star)\,w(V\setminus U^\star)
=
\frac{\varepsilon}{2}\,D^\star.
\]
We now apply Markov's inequality.
Define the random variable
$Y \triangleq D^\star - \sum_{v\in U^\star}\tilde{w}(v)$,
which is nonnegative (as in the proof of Lemma~\ref{lem:wssc_val_bound}). Hence, $\mathbb{E}[Y]
\le \left(1-\frac{\varepsilon}{2}\right)D^\star.$  By Markov's inequality,
\begin{align*}
\Pr\!\left[\sum_{v\in U^\star}\tilde{w}(v)\le \frac{\varepsilon}{4}D^\star\right]
&=
\Pr\!\left[Y \ge \left(1-\frac{\varepsilon}{4}\right)D^\star\right] \\
&\le
\frac{\mathbb{E}[Y]}{\left(1-\frac{\varepsilon}{4}\right)D^\star}\\
&\le
1-\frac{\varepsilon}{4}.
\end{align*}
It follows that:
\[
\Pr\!\left[\sum_{v\in U^\star}\tilde{w}(v)\ge \frac{\varepsilon}{4}D^\star\mid E\right]
\ge
\frac{\varepsilon}{4}.
\]
Combining this inequality with the bound $\Pr[E]\ge \frac{c}{2}\,k^{-5/2}$ and
substituting into~\eqref{eq:wssc_conditioning}, we obtain the claimed result.
\end{proof}

\begin{theorem}\label{thm:wssc_main}
For any feasible instance of the \WSSClong,
Algorithm~\ref{alg:WSSC} returns, with a probability of at least $p$,
a solution of cost at most $4\sqrt{k/\ln k}\cdot B$,
where $k = |\{v\in V : w(v)>0\}|$.
\end{theorem}
\begin{proof}
By Lemma~\ref{lem:wssc_iteration_success}, in each iteration the event
$\sum_{v\in U^\star}\tilde{w}(v) \ge \frac{1}{4}\sqrt{\ln k/k}\cdot D^\star$
occurs with a probability of at least $c/(8k^3)$.
Since the algorithm performs $\frac{8k^3}{c}\ln\!\big(\frac{1}{1-p}\big)$ iterations,
this event occurs in at least one iteration with probability at least $p$.
Condition on an iteration in which it occurs, and let $T$ be a minimizer of
$f(X)-\alpha\sum_{v\in X}\tilde{w}(v)$ in that iteration. By optimality of $T$,
\begin{align}
f(T)-\alpha\sum_{v\in T}\tilde{w}(v)
&\le
f(U^\star)-\alpha\sum_{v\in U^\star}\tilde{w}(v)\nonumber\\
&\le
f(U^\star)-\left(4\sqrt{\tfrac{k}{\ln k}}\cdot B\right)\cdot
\left(\tfrac{1}{4}\sqrt{\tfrac{\ln k}{k}}\cdot D^\star\right)\nonumber\\
&=
f(U^\star)-B D^\star
<0.\nonumber
\end{align}
Therefore, the algorithm returns $T$, and by Lemma~\ref{lem:wssc_val_bound} we get
$\tfrac{f(T)}{w(T)\cdot w(\bar T)}<\alpha=4\sqrt{k/\ln k}\cdot B$.
\end{proof}

To obtain an approximation for \WSSC, we proceed by performing a brute force search in a logarithmic scale on the value $B$ (in~\cite{svitkina2011submodular} it is stated that a binary search is used).
We provide a brief overview of the details.

First, we compute a set $\tilde{S}\subseteq V$ that is an optimal solution to:
\begin{align}
    \min _{S\subseteq V\colon \varnothing\neq S\cap K\neq K} \left\{ f(S)\right\}. \label{search1}
\end{align}
We note that~\eqref{search1} is solvable in polynomial time since it reduces to $O(k^2)$ unconstrained submodular minimization problems (one can simply enumerate over all possible pairs of elements $u,v$ in $K$, 
and solve the unconstrained submodular minimization problem $\min\{g(S)\colon S\subseteq V\setminus \{ u,v\}\}$ over $V\setminus \{ u,v\}$, where $ g\colon 2^{V\setminus \{ u,v\}}\rightarrow \mathbb{R}_{\geq 0}$ and $ g(S)\triangleq f(S\cup\{ u\})$ is submodular, using, e.g.,~\cite{GLS81}).
Denote by $\tilde{B}\triangleq f(\tilde{S})$ the value of~\eqref{search1}.

Second, let $U^{\star}\subseteq V$ be an optimal solution to the given instance of \WSSC, i.e.,
\begin{align}
    \frac{f(U^{\star})}{w(U^{\star})\cdot w(V\setminus U^{\star})}=\OPT_{\text{\WSSC}}. \nonumber
\end{align}
We note that $U^{\star}$ is a feasible solution to~\eqref{search1}, since if $U^{\star}\cap K=K$ or $ U^{\star}\cap K=\varnothing$ then $w(U^{\star})\cdot w(V\setminus U^{\star})=0$ (contradicting the fact that $U^{\star}$ is an optimal solution).
Thus,
\begin{align}
   \tilde{B}=f(\tilde{S})\leq f(U^{\star}). \label{search2}
\end{align}
Third, denote by $w_{\min}\triangleq \min \{ w(u)\colon u\in K\}$, $w_{\max}\triangleq \max \{ w(u)\colon u\in K\}$, and $W\triangleq \sum _{u\in K}w(u)$.
Therefore, for every $ S\subseteq V$ such that $\varnothing\neq S\cap K\neq K$:
\begin{align}
   w_{\min}(W-w_{\min})\leq w(S)\cdot w(\bar S)\leq \frac{W^2}{4}, \label{search3}
\end{align}
since $  w(S)\cdot w(\bar S)=x(W-x)$ for some $ w_{\min}\leq x\leq W-w_{\min}$.

To perform the brute force search in a logarithmic scale on $B$, we need to provide a suitable range in which $\OPT_{\text{\WSSC}}$ belongs.
We start with an upper bound on the range:
\begin{align}
   \OPT_{\text{\WSSC}}= \frac{f(U^{\star})}{w(U^{\star})\cdot w(V\setminus U^{\star})}\leq \frac{f(\tilde{S})}{w(\tilde{S})\cdot w(V\setminus\tilde{S})}\leq \frac{\tilde{B}}{w_{\min}(W-w_{\min})}, \label{search4}
\end{align}
where the first inequality follows from the optimality of $U^{\star}$ and the second inequality follows from the left-hand side of~\eqref{search3}.
Focusing on the lower bound of the range:
\begin{align}
   \OPT_{\text{\WSSC}}= \frac{f(U^{\star})}{w(U^{\star})\cdot w(V\setminus U^{\star})} \geq \frac{\tilde{B}}{\frac{W^2}{4}}, \label{search5}
\end{align}
where the inequality follows from combining~\eqref{search2} with the right-hand side of~\eqref{search3}.
Hence, combining the upper bound~\eqref{search4} and the lower bound~\eqref{search5}, we perform the search in the range $ [\tilde{B}/(W^2/4),\tilde{B}/(w_{\min}(W-w_{\min}))]$.

Consider an interval $[\tilde{B}/(w_{\min}(W-w_{\min}))/2^{j+1},\tilde{B}/(w_{\min}(W-w_{\min}))/2^j]$ for $j=0,\ldots,L$, where $L=O(\ln{(W/w_{\min})})=O(\ln{(n\cdot w_{\max}/w_{\min}}))$ (note that $ W\leq n\cdot w_{\max}$) such that the entire range is covered.
For each such $j$, we invoke Algorithm~\ref{alg:WSSC} with $B=\tilde{B}/(w_{\min}(W-w_{\min}))/2^j$, setting the failure probability to be $ (1-p)/L$.
There must exist an index $j^{\star}$ for which $ \tilde{B}/(w_{\min}(W-w_{\min}))/2^{j^{\star}+1}\leq \OPT_{\text{\WSSC}}\leq \tilde{B}/(w_{\min}(W-w_{\min}))/2^{j^{\star}}$, and for the invocation corresponding to this $j^{\star}$, Theorem~\ref{thm:wssc_main} guarantees a solution whose cost is at most $ O(\sqrt{k/\ln{k}})\cdot \OPT_{\text{\WSSC}}$.
Taking a union bound of the failure probability over all $L$ invocations, implies that we succeed with an overall probability of at least $p$ and the running time of each invocation is polynomial in $n$ and an additional factor of $ O(\ln{(L/(1-p))})$.
We conclude since $ L=O(\ln{(n\cdot w_{\max}/w_{\min})})$ is polynomial in $n$ and the number of bits required to represent the weight function $w$.

The following corollary summarizes the guarantee for \WSSC.
\begin{corollary}\label{cor:wssc_approx}
There exists a polynomial-time randomized algorithm that, given a submodular
set function $f$ and a nonnegative weight function $w$ on $V$, where $\lvert V \rvert=n$, with a probability
of at least $p$, where $p$ satisfies that $ 1-p\geq \Omega(\text{exp}(-\text{poly}(n)))$, returns a cut $T \subseteq V$ satisfying:
\[
\frac{f(T)}{w(T)\cdot w(\bar T)}
\le
O\!\left(\sqrt{\frac{k}{\ln k}}\right)\cdot
\operatorname{OPT}_{\mathrm{WSSC}},
\]
where $k = |\{v\in V : w(v)>0\}|$.
\end{corollary}

\paragraph{Proving Theorem~\ref{thm:balanced_cut_algorithm}.}
We now analyze the \SBC algorithm given by Algorithm~\ref{alg:SBC_general}, provided the \WSSC guarantee of Corollary~\ref{cor:wssc_approx}.
At a high level, the algorithm repeatedly applies the WSSC subroutine and peels off one side of the resulting cut by setting its weights to zero.
In each iteration, if the lighter side is $S_i$ it is added to $T_1$, and otherwise
$\bar S_i$ is added to $T_2$; this controlled removal guarantees that the final output
is a balanced cut, while allowing $f(T_1)$ and $f(\bar T_2)$ to be bounded using the guarantee of Corollary~\ref{cor:wssc_approx}. 

\begin{algorithm}[t]
\caption{\SBClong}
\label{alg:SBC_general}
\begin{algorithmic}[1]
\Require $V,f,w$, and parameter $b'$
\Ensure a set $T \subseteq V$

\State $w' \gets w$, \quad $i \gets 0$, \quad $T_1 \gets \varnothing$, \quad $T_2 \gets \varnothing$
\While{$w'(V) > (1-b')\,w(V)$}
    \State invoke the algorithm from Corollary~\ref{cor:wssc_approx} on $(V,f,w')$ to obtain $S_i\subseteq V$
    \If{$w'(S_i) \le w'(\bar S_i)$}
        \State $T_1 \gets T_1 \cup S_i$
        \State set $w'(v)\gets 0$ for all $v\in S_i$
    \Else
        \State $T_2 \gets T_2 \cup (\bar S_i)$
        \State set $w'(v)\gets 0$ for all $v\in (\bar S_i)$
    \EndIf
    \State $i \gets i+1$
\EndWhile
\State \textbf{if} $w(T_1) \ge w(T_2)$ \textbf{then return} $T_1$ \textbf{else return} $\bar T_2$
\end{algorithmic}
\end{algorithm}

\begin{proof}[Proof of Theorem~\ref{thm:balanced_cut_algorithm}]
In each iteration $i$, the \textsc{WSSC} subroutine returns a cut $S_i$ with
$w'(S_i),w'(\bar S_i)>0$ (otherwise its \textsc{WSSC} value is infinite); hence in
each iteration we set $w'(v)=0$ for at least one previously positive-weight element,
and the loop terminates after at most $n$ iterations.

When the while-loop terminates we have $w'(V)\le (1-b')\,w(V)$, and since $w'$
is obtained from $w$ by setting weights to zero on the elements added to
$T_1\cup T_2$, it follows that $w(T_1\cup T_2)\ge b'\,w(V)$.
Therefore $\max\{w(T_1),w(T_2)\}\ge (b'/2)\,w(V)$.
Now consider the last iteration of the loop. At its beginning,
$w'(V)>(1-b')\,w(V)$, and in this iteration we set to zero the lighter of
$S_i,\bar S_i$, hence the total remaining weight decreases by at most
$\min\{w'(S_i),w'(\bar S_i)\}\le \sfrac{1}{2}\, w'(V)$. Consequently, at the end of
this iteration we still have $w'(V)>\sfrac{1}{2}\,(1-b')\,w(V)$.
Now fix $T\in\{T_1,T_2\}$. Since $w'(T)=0$, we have $w'(\bar T)=w'(V)$ and thus
$w(\bar T)\ge w'(\bar T)=w'(V)>\sfrac{1}{2}\,w(V)\ge (b'/2)\,w(V)$,
where the last inequality uses $b'\le \sfrac{1}{2}$.
We have shown that for $T\in\arg\max\{w(T_1),w(T_2)\}$ it holds that
$w(T),w(\bar T)\ge(b'/2)\,w(V)$, and thus the algorithm returns a
$b'/2$-balanced cut.

Let $Q^\star$ be an optimal $b$-balanced cut with respect to the original weights $w$, i.e.,
  \[
f(Q^\star)
=
\min\bigl\{ f(A)\ \big|\ A\subseteq V,\; b\,w(V)\le w(A)\le (1-b)\,w(V) \bigr\}.
\]

Fix an iteration $i$. Since the loop condition holds, $w'(Q^\star)+w'(\bar Q^\star)>(1-b')\,w(V)$.
Moreover, since $Q^\star$ is $b$-balanced with respect to $w$, we have
$\max\{w'(Q^\star),w'(\bar Q^\star)\}\le (1-b)\,w(V)$. It follows that
$w'(Q^\star)\cdot w'(\bar Q^\star)\ge (b-b')\,w(V)\cdot(1-b)\,w(V)$.
By Corollary~\ref{cor:wssc_approx}, the set $S_i$ is an approximate solution to the WSSC instance with respect to the current weights $w'$ in iteration $i$.
Since the approximation factor is at most $\gamma = O\!\left(\sqrt{k/\ln k}\right)$ and $Q^\star$ is a feasible solution to this instance, we have:
\[
\frac{f(S_i)}{w'(S_i)\,w'(\bar S_i)}
\;\le\;
\gamma \cdot
\frac{f(Q^\star)}{w'(Q^\star)\,w'(\bar Q^\star)}
\;\le\;
\frac{\gamma\,f(Q^\star)}{(b-b')\,w(V)\cdot(1-b)\,w(V)}.
\]
Let $P_i\in\{S_i,\bar S_i\}$ denote the side whose weights are set to zero in
iteration~$i$ (the lighter side). Since $w'(P_i)=w\!\left(P_i\setminus\bigcup_{j=0}^{i-1} P_j\right)$,
$w'(S_i)\,w'(\bar S_i)\le w'(P_i)\,w(V)$, and $(1-b)\ge 1/2$, it follows that:

\[
f(S_i)
\;\le\;
w\!\left(P_i\setminus\bigcup_{j=0}^{i-1} P_j\right)\cdot
\frac{2\gamma\, f(Q^\star)}{(b-b')\,w(V)}.
\]
To complete the proof, we bound the values of $f(T_1)$ and $f(\bar T_2)$.
By submodularity and nonnegativity,
\[
f(T_1)
\;\le\;
\sum_{i:\,P_i=S_i} f(S_i)
\;\le\;
w(T_1)\cdot\frac{2\gamma\, f(Q^\star)}{(b-b')\,w(V)}
\;\le\;
O\!\left(\frac{\gamma}{b-b'}\right)\cdot f(Q^\star).
\]
An analogous argument applied to $\bar f(S)\triangleq f(\bar S)$ (which is
known to be submodular) yields:
$f(\bar T_2)=\bar f(T_2)\le \sum_{i:\,P_i=\bar S_i} f(S_i)$,
and the same bound follows.
\end{proof}

\paragraph{Probability Amplification.}
The probability amplification for \SBC follows exactly the same argument used
above for \WSSC, since the \SBC algorithm makes $O(n)$ calls to the \WSSC
subroutine, and can therefore be made to succeed with any desired probability $p$ satisfying $1-p\geq \Omega(\exp(-\text{poly}(n)))$.
Similarly, the \MLOP algorithm makes $O(n)$ calls to \SBC, and the same
amplification applies.
Consequently, the \MLOP algorithm runs in polynomial time and can be made to
succeed with an arbitrarily high probability $p$ satisfying $1-p\geq \Omega(\exp(-\text{poly}(n)))$.

\end{document}